\documentclass[reqno,12pt]{article}
\usepackage{amsmath}
\usepackage{bm}
\usepackage{amssymb}
\usepackage{graphicx,
}
\usepackage{epstopdf}
\usepackage{amsthm}
\usepackage{float}
\usepackage{color}

\newcommand{\la}{\lambda}

\def\s{\sigma}

\def\l{{\lambda}}

\def\CC{{\cal C}}

\def\th{\theta}

\newtheorem{oss}{Remark}

\newtheorem{lem}{Lemma}
\newtheorem{prop}{Proposition}

\newtheorem{defi}{Definition}

\newtheorem*{proof*}{Proof}

\def\be{\begin{equation}}
\def\ee{\end{equation}}
\def\bea{\begin{eqnarray}}
\def\eea{\end{eqnarray}}

\def\nn{\nonumber}

\numberwithin{equation}{section}
\numberwithin{thm}{section}

\begin{document}

\title{Equilibria of a clamped Euler beam (\textit{Elastica}) with distributed load: large deformations.}

\author{Della Corte A.$^{1,2}$, dell'Isola F.$^{1,3}$, Esposito R.$^{1}$ and Pulvirenti M.$^{1,4}$}

\date{}

\maketitle

\small{\noindent 1: International Research Center on the Mathematics and Mechanics of Complex Systems M\&MoCS, University of L'Aquila, Italy.

\noindent 2: Department of Mechanical and Aerospace Engineering, Sapienza University of Rome, Italy. 

\noindent 3: Department of Structural and Geotechnical Engineering, Sapienza University of Rome, Italy.

\noindent 4: Department of Mathematics, Sapienza University of Rome, Italy.}

\begin{abstract}
\noindent We present some novel equilibrium shapes of a clamped Euler beam (\textit{Elastica} from now on) under uniformly distributed dead load orthogonal to the straight reference configuration. We characterize the properties of the minimizers of total energy, determine the corresponding Euler-Lagrange conditions and prove, by means of direct methods of calculus of variations, the existence of curled local minimizers. Moreover, we prove some sufficient conditions for stability and instability of  solutions of the Euler-Lagrange, that can be applied to numerically found curled shapes.
\\


\end{abstract}

\section{Introduction}

In the present paper we consider the planar configurations of an inextensible \textit{Elastica}\begin{footnote}[1] {Euler, 1744 \cite{Euler1}}\end{footnote} of length $L$, clamped at one of its ends, which is assumed to be the origin of the reference curvilinear abscissa $s$. The end point corresponding to $s=L$ is assumed to be free. 
\noindent Let us assume that $\mathcal{E}^2$ is the affine euclidean plane including the placements of the \textit{Elastica}. In $\mathcal{E}^2$ we introduce a coordinate system whose origin $\mathcal{O}$ coincides with the reference placement of the clamped end of the \textit{Elastica}, and whose basis $(\boldsymbol{D}_1, \boldsymbol{D}_2)$ is orthonormal. $\boldsymbol{D}_1$ is the tangent unit vector to the reference shape of the \textit{Elastica} (assumed to be straight when undeformed) and coincides with the imposed clamping direction at the clamped end. 

\subsection{Some considerations on the history of the problem}

The theory of large deformations of \textit{Elasticae} was fully formulated by Euler already in 1744 \cite{Euler1}, with the important contribution of the ideas of Daniel \cite{DBernoulli} and James Bernoulli \cite{JBernoulli}.
One important tool for finding the equilibrium forms of \textit{Elasticae} was developed by Lagrange \cite{Lagrange}, who formulated the stationarity condition for their total energy.  The relative boundary value problems for the resulting ordinary differential equations have been since then extensively used for determining the equilibrium shapes. Without even trying to provide a complete bibliography, it is worth mentioning that large deformations of the \textit{Elastica} have been studied by many other great scientists (including for instance Max Born in his Doctoral Thesis \cite{Born}).  Euler's beam model, in the inextensible case considered herein, has been validated by means of rigorous derivation from 3D elasticity \cite{Mora,Pideri2} and systematically and thoroughly used in structural mechanics.

\noindent The theory of \textit{Elastica}, indeed, still represents the basis for many more complex problems of structural mechanics (as for instance the behavior of the fibrous structures described in~\cite{Rivlin,Pipkin}). In the literature we could not find rigorous results on the study of the whole set of equilibrium shapes and of their stability for an \textit{Elastica} in large deformation under a uniformly distributed load. There were, however, some suggestive numerical results that deserved, in the opinion of the authors, greater attention than they received. Namely, in \cite{Faulkner,Raboud} some numerical results producing equilibrium shapes similar to the ones we will present later are shown; in these papers the solutions of the boundary value problem are obtained by means of a shooting technique. In the older paper \cite{Fried}, some numerically evaluated stability results on shapes similar to those considered herein are presented in case of concentrated end load. The cited papers appeared when the first effective numerical codes capable to deal with nonlinear boundary value problems became available. In subsequent investigations, while the numerical methods were more and more developed, the mechanical problems related to large deformations of beams, with the noticeable exception of problems with concentrated end load (see \cite{Antman} for a rational bibliography on the problem), seem to have somewhat escaped rigorous treatment. 

\noindent In the recent past, however, the awareness of the importance of large deformations problems in structural mechanics came back in both theoretical~\cite{Drucker,Fertis,Forest,Abali} and computational \cite{Ladeveze} directions, and this importance will probably increase whether further substantial progresses will be achieved, in particular since large deformations play a relevant role in topical research lines as the design of metamaterials~\cite{RoyalIVAN,Milton,AMR,Pideri}. An interesting research stream in the design of novel metamaterials involved the so-called pantographic structures \cite{Rizzi,Turco,ZAMP}, which were first conceived in order to synthesize a particular class of generalized continua, \textit{i.e.} those described by second gradient energy models \cite{Mindlin,Germain}. Exactly as in many other conceivable metamaterials, pantographic ones base their exotic behavior on the particular geometrical and mechanical micro-structure of beam lattices \cite{Alibert} and on the deformation energy localization allowed by the onset of large deformations in portion of  beams located in some specific areas of the lattice structure. It is therefore clear that, if one is interested in designing and optimizing pantographic metamaterials, a reasonably complete knowledge of the behavior of beams in large deformations is needed. The authors were indeed pushed in the present research direction exactly pursuing the aforementioned aims.

\subsection{Setting of the problem}

\noindent The configurations of the system are curves $$\boldsymbol{\chi}: s \in [0,L] \rightarrow \boldsymbol{\chi}(s) \in \mathcal{E}^2=\chi_1(s)\boldsymbol{D}_1+\chi_2(s)\boldsymbol{D}_2$$ where $\mathcal{E}^2$ is the Euclidean plane including the placement of the \textit{Elastica}.
We assume that the reference configuration is given by $$\boldsymbol{\chi}_0:s \in [0,L] \rightarrow \boldsymbol{\chi}_0 (s)=s\boldsymbol{D}_1\in\mathcal{E}^2$$ The considered \textit{Elastica} is subject to an inextensibility constraint, so that placements verify the local condition (denoting with apexes differentiation with respect to the reference abscissa): 
\be
\label{inext}
\boldsymbol{\chi}' (s)\cdot \boldsymbol{\chi}' (s)=1 \quad \text {for all} \quad s\in [0,L].
\ee
The elastic energy of the inextensible \textit{Elastica}, as originally assumed by Bernoulli and Euler, is given by the quadratic form
$$
\frac 12 \int_0^L k_M \boldsymbol{\chi}'' (s)\cdot \boldsymbol{\chi}'' (s)ds=\frac{1}{2}\int_0^Lk_M\kappa^2(s)ds,
$$
where $k_M$ is the bending stiffness of the \textit{Elastica} (we assume $k_M'(s)=0$) and $\kappa(s)$ is the geometrical curvature of the actual shape of the \textit{Elastica} (because of the inextensibility condition).  

\noindent Finally we will suppose that the \textit{Elastica} is subject to the uniformly distributed dead load (like gravitational force) given by $\textbf{b}(s)=b\boldsymbol{D}_2$, with $b>0$, whose associated energy is
$-b\int_0^L \chi_2 (s) ds$. Therefore the equilibria of the system are the stationary points of the energy functional
\be
\label{energy}
E(\chi(\cdot))= \frac 12 \int_0^L \boldsymbol{\chi}'' (s)\cdot \boldsymbol{\chi}'' (s) ds- b \int_0^L  \chi_2 (s) ds,
\ee
with the additional condition \eqref{inext}. The left boundary conditions are $\boldsymbol{\chi}(0)=\boldsymbol{\chi}'(0)=\mathbf{D}_1$, while the right boundary conditions (at $s=L$) are free.

In this paper we study the equilibrium configurations of the system, namely the stationary points of $E$.

We treat the constraint \eqref{inext} by introducing a different configuration field verifying the inextensiblity condition automatically. We define the angle $\theta$, counted in the counter-clockwise sense by the position
\be
\boldsymbol{\chi}'(s)=\cos\theta(s) \boldsymbol{D}_1+\sin\theta(s) \boldsymbol{D}_2.
\ee
Once the scalar field $\theta(s)$ is known, together with the clamping condition  $\boldsymbol{\chi}(0)=\mathcal{O}$, the placement $\boldsymbol{\chi}(s)$ is uniquely determined by integration. The condition imposing the clamping direction is equivalent to $\theta(0)=0$.
\subsection{Description of the results}
Once formulated the problem which consists in studying the stationary points of the enegy functional expressed in terms of $\th$, we make use of the well known direct method of calculus of variations. The existence of global and local minima of the energy functional is ensured by general theorems based on results first obtained by L. Tonelli \cite{Tonelli,Dacorogna}.

A qualitative analysis of the shape of these stationary points is carried out by means of comparisons based on energy considerations. 

In this way we find two branches of solutions, parametrized by the load parameter. The first, called ``primary branch'', corresponds to global minima of the energy. The solutions, characterized by a positive rotation $\th$, look like equilibria of a trampoline under the action of a positive gravitational field (see Fig. 8 below).

Another branch of solutions, called ``secondary branch'' is a topological extension of a family of solution with a negative angle $\th$ (see Fig. \ref{Stable} below) rotating around the origin (clamping point), which are local minimizers of the energy. These solutions are possible only if the load is sufficiently large or, equivalently, the beam is sufficiently long.


Profiles of the primary and secondary branches are found also numerically by means of a shooting technique. This complements previous qualitative analysis by further quantitative information.

In order to establish a suitable stability chart of these solutions, we note that  the primary branch is obviously stable since it is formed by global minima. As for the secondary branch, the situation is richer. Indeed, by establishing sufficient conditions for the stability/instability of the solutions of the Euler-Lagrange equations associated to our variational problem we are able to show that the secondary branch contains stable solutions and give sufficient conditions for the instability.

\vskip .3cm
\noindent The paper is organized as follows: in Section 2 we reformulate the problem of equilibrium of \textit{Elastica} under uniform load in terms of the rotation field of the cross-section. This allows for a  Lagrangian characterization of \textit{Elastica}'s space of configurations. In Section 3 we provide rigorous results on the characterization of the equilibrium shapes attaining global energy minima when the intensity of externally applied load varies. In Section 4 we prove the existence of a branch of stable equilibrium shapes exhibiting a curling around the clamped extremum. In Section 5 we provide some results on the stability and instability of stationary points verifying some particular properties. In Section 6 we apply these results to numerically found stationary points. In Section 7 we show that the ending part of the curled equilibrium shape coincides (up to a reflection) to the shape of the global minimizer for an \textit{Elastica} of suitably reduced length.

\section{Reformulation of the problem and first considerations}

As already explained, in order to eliminate the constraint \eqref{inext} it is convenient to introduce the new variable $\theta$ that is the angle (counted in the anti-clockwise sense) formed by the
unitary tangent $\boldsymbol{\chi}' (s)$ to the graph of the curve $\boldsymbol{\chi}$ with the $x$- axis. In the new configuration field the energy reads\be
\label{energy1}
E(\theta(\cdot))= \frac 12 \int_0^L |\th' (s)|^2 \text{d}s - b \int_0^L  (L-s) \sin \th (s) \text{d}s.
\ee

\begin{oss} 
\label{remark0}
Since $|\sin\th|\le 1$, it is immediate to see that $E(\th)\ge-\frac{b}{2}$ for all admissible $\th\in L^2([0,L])$. In particular, $E$ is bounded from below.
\end{oss}
\noindent A standard formal computation shows that the Euler-Lagrange conditions
associated to the functional \eqref{energy1} are
\be
\label{1}
\th''=-b (L-s) \cos \th, 
\ee
in the set of admissible functions verifying

\be
\label{bc}
\th(0)=\th'(L)=0.
\ee
Note that, while the condition
\be 
\label{bcbis}
\theta(0)=0, 
\ee
characterizes the admissible kinematics and is given by the problem, the condition at $L$
is a consequence of imposing stationarity, and can be interpreted by saying that at the free end 
the curvature must vanish. 

We remark that also the boundary condition $\th(0)=\pi$ would also describe a clamped beam. This boundary condition would give rise to solutions are specular reflection the solutions with boundary data \eqref{bcbis}. Therefore we will ignore them.

By obvious scaling properties of equation \eqref {1} we see that any solution of such an equation, for a given pair $L$ and $b$,
can be recovered by setting $L=1$ rescaling $b \to bL^3$. Therefore from now on we shall assume $L=1$ leaving
$b$ as the only parameter.

\noindent By \eqref{1} we easily obtain the following integral equation for the unknown field $\th$
\be \label{eqint}
\th(s) =b \int _0^s d\s \int_{\s}^1 (1- \s') \cos \th (\s') d \s''.
\ee
We define the nonlinear integral operator $T$ as 
\be \label{eqint1}
[T\theta](s)=b \int _0^s d\s \int_{\s}^1 (1- \s') \cos \th (\s') d \s' 
\ee
The solutions to \eqref{1} with conditions \eqref{bc} and \eqref{bcbis} are all the fixed points of the map $ \th \to T (\th)$. 

\begin{prop}
Under the smallness assumption
\be\frac {b} 6 <1.\label{smallness}
\ee
the solution of the problem \eqref{1}, \eqref{bc}, \eqref{bcbis} exists and is unique. 
\label{uniqueness_}
\end{prop}

\begin{proof}
It is easy to check that under the condition \eqref{smallness}, the operator $T$ is a strict contraction and hence the proof follows by  a standard application of the Banach fixed point theorem.
\end{proof}

To understand the behavior of the solution obtained for small $b$,
we compute the solution of the boundary value problem
$$
\th''=-b (1-s), 
$$
$$
\th(0)=\th'(1)=0,
$$
which follows by setting $\cos \th =1$,  which is a reasonable approximation when $b$ is small.
The result is
\be 
\label{bsmall}
\th (s)= \frac {-b}{6} [ (1-s)^3-1],
\ee
which means that $\th$ is strictly positive, increasing and concave. 

\noindent To go further, namely looking for solutions with a large $b$, we investigate the minima of the energy functional \eqref{energy1}.
We first observe that the energy has the form
\be
\label{energy2}
E(\th(\cdot))=\int_0^1 ds \, {\cal L } ( \th', \th, s),
\ee
where
\be
{\cal L } (\theta',\theta, s)=\frac 12 (\theta'^2 -b  (1-s) \sin \theta).
\ee
We define $E$  on the space 
\be
H^1_{0,1}= \{ \th \in H^1 [0,1] |  \th(0)=0 \}.\label{H101}
\ee 

\noindent Note that
$$
p \to {\cal L } ( p,q, s),
$$
is a continuous function which is convex and coercive for all $q$ and $s$. Therefore, by a classical  result in Calculus of Variations (see for instance \cite{DeFonseca}, Section 3.2), we have a minimizer in $H^1_{0,1}$, not necessarily unique.
Such a minimizer, say $\th$, is a solution of \eqref {1} in a weak form \cite{DeFonseca,Evans}, for any value of $b$, namely

\be
\label{weak}
(\varphi', \th')=b (\varphi,  (1-s) \cos \th),
\ee
for all $\varphi \in H^1_{0,1} $ vanishing at the extrema. 

\begin{lem}
\label{lemma1}
If $\theta$ is weak solution of equation \eqref{weak}, then it is in $C^{\infty}$ and is a classical solution of \eqref{1}. In particular, any stationary point of the energy is a classical solution of the problem \eqref{1}, \eqref{bc}.
\end{lem}

\noindent \textit{Proof}
Let us  define $\Psi=T\th $. Straightforward computation provides

\be
\Psi''=-b(1-s)\cos\theta(s),
\ee
and therefore, by equation \eqref{weak}, we get 

$$
(\varphi', \th')=-(\varphi,\Psi'')=(\varphi', \Psi'),
$$
\noindent so that $$\th=\Psi.$$ 
\noindent By definition of the operator $T$ (see \ref{eqint1}), $\Psi$ is in $C^2$ if $\theta$ is $H^1$. Hence the previous equality tells us that  $\th$ is in $C^2$ as well, and applying the argument recursively one obtains that $\th$ is in $C^{\infty}$. 

\qed

\section{The primary branch}
\noindent In this and in the next section we investigate some properties of stationary points of $E(\theta(\cdot))$ and in particular of its minimizers.

\subsection{Qualitative properties of the minimizer}
We start by proving some properties of the stationary points introduced in the previous Section.

\begin{prop}
\label{prop1}
For any $b>0$ there exists a minimizer $\tilde{\th} \in C^{\infty} [0,1]$ of the energy $E$, strictly increasing and concave.
Moreover $\tilde{\th}(1) < \frac {\pi}2$.  
\end{prop}
The statement of the above proposition is illustrated in Fig. 1.
\begin{figure}[H]
\centering
\label{fig1}
\includegraphics[width=0.5\textwidth]{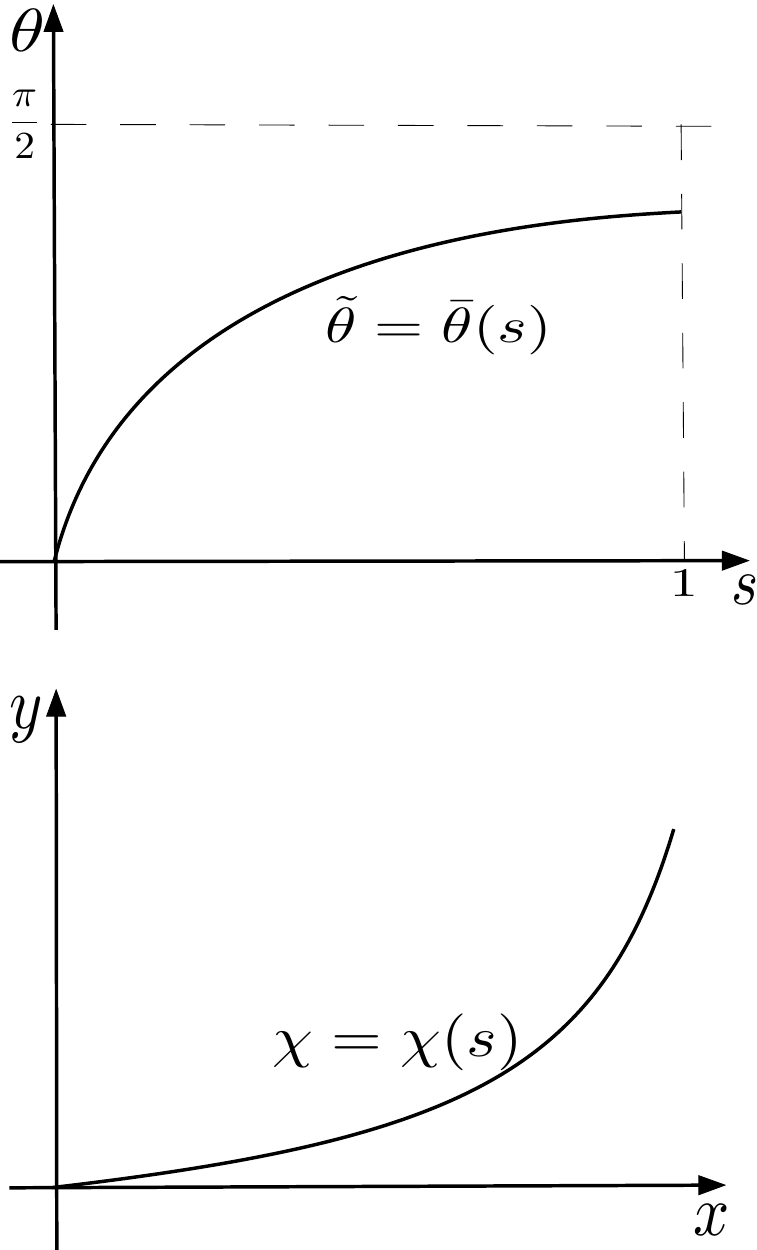}
\caption{$\theta$-graph of the minimizer  and corresponding shape of the beam.}
\end{figure}


\begin{proof}
\normalfont 
Let $\tilde{\th}$ be a minimizer of $E$. Then it is a solution of \eqref {weak} and, because of Lemma \ref{lemma1}, of \eqref{1}. 

We note that $\tilde{\th}''(0)\leq 0$. Then $\tilde{\th}(s)$ is concave for $s$ sufficiently small and, by using the equation,
$\tilde{\th} $ must cross either $\pi/2$ or $-\pi/2$ to change its concavity. 

We start by proving that  $\tilde{\th} $ is increasing in a sufficiently small right neighborhood of the origin. Supposing the contrary by contradiction, then $\tilde{\th}$ starts to decrease. Since $\tilde{\th}'(1)=0$, $\tilde{\th} $ must change concavity and this means it has to cross $-\pi/2$ in a point $s_0$.
\noindent Consider now the function $\xi \in H^1_{0,1}$ defined as
$$
\xi (s) = - \tilde{\th}(s) \quad \text {for} \quad s \in (0, s_0 ), \quad  \xi (s)= \frac {\pi} 2 \quad \text {for } \quad s \geq s_0.
$$ 
Then
$$
E(\xi)  < E( \tilde{\th}).
$$
Indeed:
$$
\int_0^{1} |\xi '(s)|^2 ds \leq \int_0^{1} |\tilde{\th} '(s)|^2  ds, 
$$
and, as $\tilde{\th}<0$ and the function $\sin(\,\cdot\,)$ has the same sign of its argument in $(-\frac{\pi}{2})$, it follows

$$
\int_0^{1} (1-s)  \sin \xi (s) ds > \int_0^{1} (1-s) \sin \tilde{\th} (s) ds.
$$
The situation is represented in Fig. 2.
\begin{figure}[H]
\centering
\label{fig2}
\includegraphics[width=0.5\textwidth]{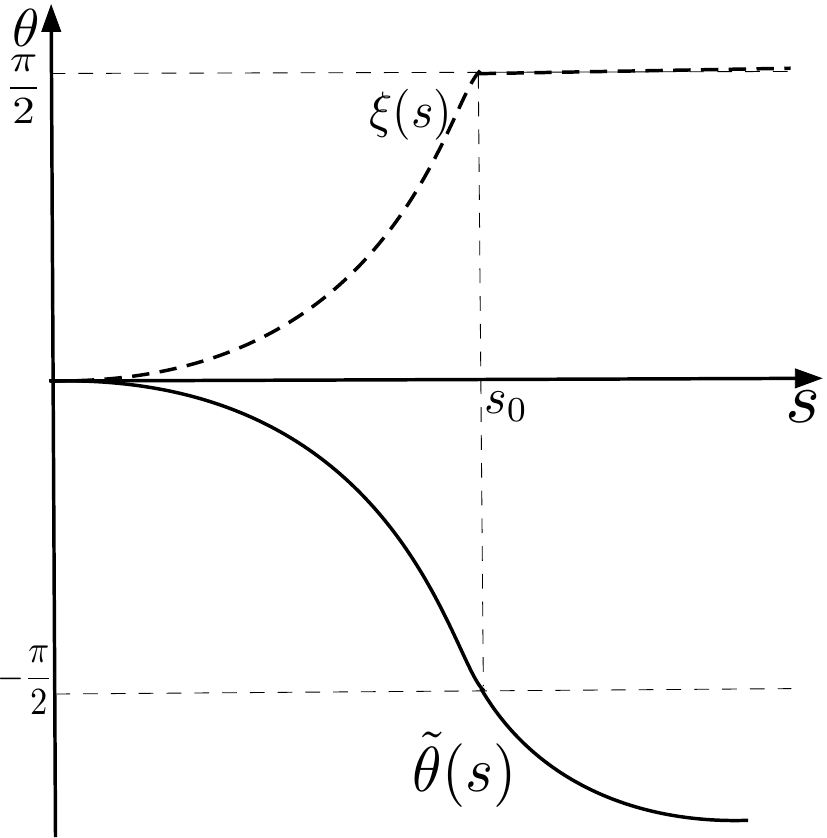}
\caption{$\tilde\th$ continuous line; $\xi$ dashed line.}
\end{figure}

Next we observe that  $\tilde{\th} (s)$ cannot cross the axis ${\th} (s)=\pi/2$ at a point $s_0$ (therefore changing concavity at that point)
because the function (see Fig. 3)
\begin{figure}[H]
\centering
\label{fig3}
\includegraphics[width=0.5\textwidth]{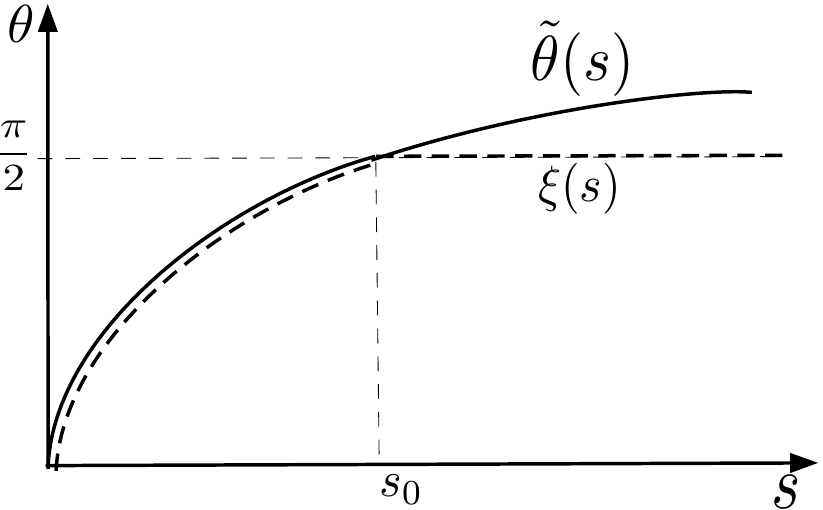}
\caption{$\tilde\th$ continuous line; $\xi$ dashed line.}
\end{figure}

$$
\xi (s) =  \tilde{\th}(s) \quad \text {for} \quad s \in (0, s_0 ), \quad  \xi (s)= \frac {\pi} 2 \quad \text {for } \quad s \geq s_0.
$$ 
is again energetically more convenient.

Suppose now that ${\th} (s)$ has a maximum below $\pi/2$ after which it decreases to cross the axis $- \pi/2$. This 
 change of concavity (at $s_0$) is necessary to 
satisfy the condition $\tilde {\th}' (1)=0$. Consider now the function
$$
\xi (s) = \tilde{\th}(s) \quad \text {for} \quad s \in (0, s_1 ), \quad \xi (s) = -\tilde{\th}(s) \quad \text {for} \quad s \in (s_1, s_0 ), 
$$
$$
\xi (s)= \frac {\pi} 2 \quad \text {for } \quad s \geq s_0,
$$ 
where $s_1$ is the first point in which $\tilde{\th}$ vanishes. See Fig. 4. Then the same arguments as before show that
the profile $\xi$ is energetically more convenient.
\begin{figure}[H]
\centering
\label{fig4}
\includegraphics[width=0.5\textwidth]{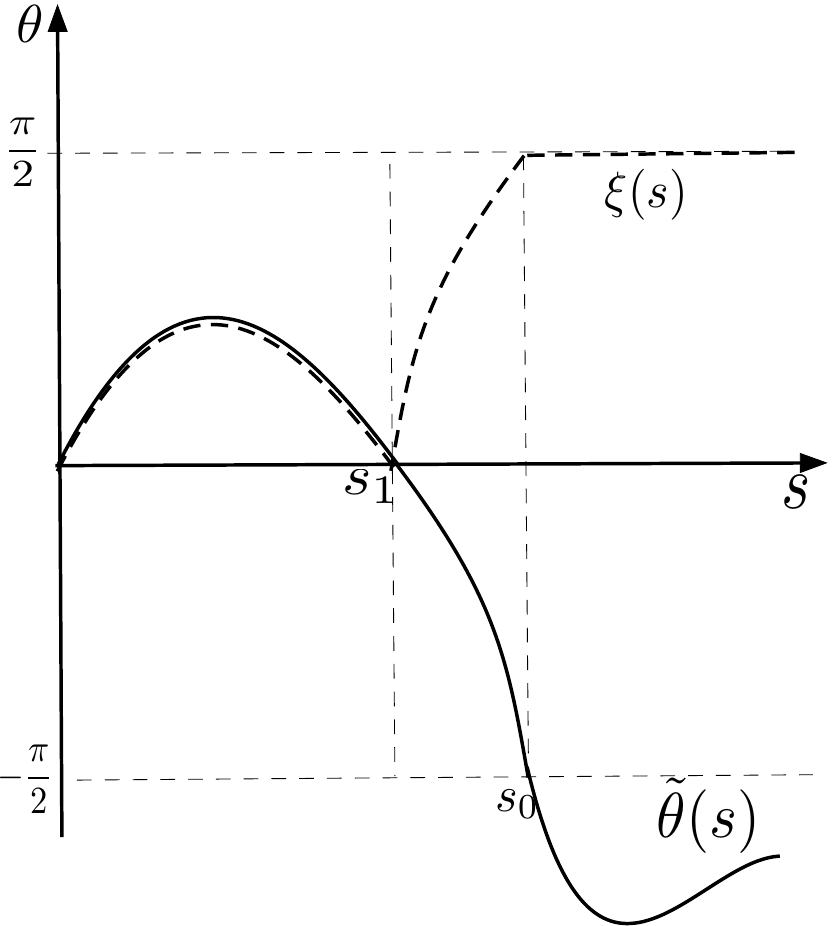}
\caption{$\tilde\th$ continuous line; $\xi$ dashed line.}
\end{figure}
It remains to exclude the case depicted in Fig. 5,
namely when $\tilde{\th}$ reaches the value $\pi/2$ at $s_0 <1$ and then it 
proceeds constantly.  However this situation is excluded.
\begin{figure}[H]
\centering
\label{fig5}
\includegraphics[width=0.5\textwidth]{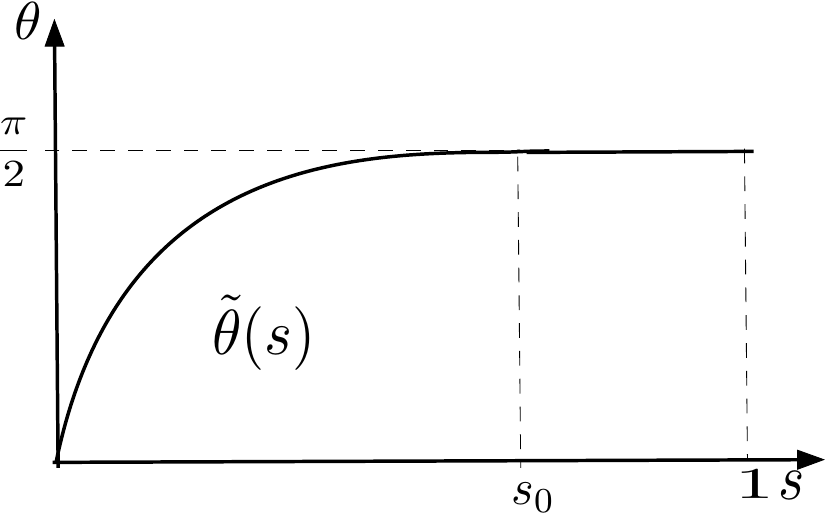}
\caption{$\tilde{\th}$ reaches the value $\pi/2$ at $s_0 <1$}
\end{figure}
\noindent Suppose indeed that  $\tilde{\th}'(s)=0$  for $s>s_0$ and $\tilde{\th}$ is not constant for $s<s_0$.
Consider now the Cauchy problem: 
\be
\th''=-b (1-s) \cos \th, 
\ee
\be
\th(s_0)=\frac{\pi}{2},\quad \th'(s_0)=0.
\ee
Then the solution $\th=\frac{\pi}{2}$ is unique in a neighbour of $s_0$, which contradicts the fact that $\tilde{\th}$ is
not constant for in a left neighbourh of $s_0$. 
By a similar argument we show that $\tilde\theta(1)<\frac{\pi}2$. We already know that $\tilde\theta(1)\le\frac{\pi}2$. By contradiction, suppose that $\theta(1)=\frac{\pi}2$ and look at the solution of the following problem:
\be
\th''=-b (1-s) \cos \th, 
\ee
\be
\th(1)=\frac{\pi}{2},\quad \quad\th'(1)=0.
\ee
The solution $\theta(s)=\frac{\pi}{2}$ is unique in a left neighbour of 1.  Moreover we know that $\tilde{\th}\ne\frac{\pi}{2}$ in a left neighbour of 1, and therefore it cannot be that $\tilde{\th}(1)=\frac{\pi}{2}$.
By the monotonicity of $\tilde\theta(\,\cdot\,)$ we also conclude that $\tilde\theta(s)<\frac{\pi}2$ for any $s\in [0,1]$.
Then, by equation \eqref{1} we also conclude that $\theta"(s)<0$ for any $s\in [0,1]$ and hence $\tilde\theta$ is strictly convex. This completes the proof of Proposition 1. 
\end{proof}

\noindent We now establish the uniqueness of the minimizer of the energy functional~\eqref{energy}.

\begin{prop}
 For any $b>0$ let $\bar \th$ be a solution to eq.n \eqref {1} with range contained in $[0,\frac {\pi}2)$.
Then $\bar \th= \tilde{\th}$. In particular there is a unique minimizer of $E(\theta(\cdot))$ and no other stationary points with range in $[0, \frac{\pi}2)$.
\end{prop}
\bigskip

\bigskip

\begin{proof}
\normalfont The set $\mathcal{S}$ of functions $\theta$ defined on $[0,1]$ and such that $0\le\theta(s)<\frac{\pi}{2}$ is a convex set. The second variation of functional \eqref{energy} in $\th$ is 

\begin{equation}
\int_0^1[(h')^2+b(1-s)\sin\theta(s)h^2]ds,
\label{secondvar}
\end{equation}
and it is clearly positive in $\mathcal{S}$ as $\sin \th \geq 0$. Suppose now that $\tilde{\theta}\ne\bar{\th}$. Consider now the function
$$
E(\lambda):\l \in [0,1] \to E ( \l \bar \th + (1-\l) \tilde{\th} ).
$$
Since $\bar{\th}$ solves \eqref{1} it is a stationary point of \eqref{energy}. However, the function $E(\lambda)$ is convex and has a minimum in $0$. Therefore, it cannot have a stationary point in $1$, which contradicts the hypothesis that $\tilde{\theta}\ne\bar{\th}$. 
\end{proof}

\noindent We want now to study the topology of the set of solutions when $b$ changes. To this aim let us begin with a standard definition.

\begin{defi} Let us denote by $\theta_{\bar{b}}$ a solution of equation \eqref{1} with $b=\bar{b}$, $\theta(0)=0$ and $\theta'(1)=0$. Given an arbitrary $B>0$, we say that the map $ b \in [0,B] \to \th_b \in H^1_{0,1}$ is a branch of solutions if it is continuous as a function of $b$. 
\end{defi}

\noindent We first establish the following:

\begin{lem}
\label{energymonotone}
Let $\th_b$ denote the minimizer of the energy (expliciting the dependence on $b$). The function $b \rightarrow E_b(\th_b)$ is decreasing.
\label{decreasing_}
\end{lem}

\noindent \textit{Proof}
Let be $b_2>b_1$. Then we have: 
\be
E_{b_2}(\th_{b_2})-E_{b_1}(\th_{b_1})=E_{b_2}(\th_{b_2})-E_{b_2}(\th_{b_1})+E_{b_2}(\th_{b_1})-E_{b_1}(\th_{b_1}). 
\ee
The last member is negative since $\th_{b_2}$ is the minimizer of $E_{b_2}$ (which makes negative the difference between the first two terms) and by direct substitution in \eqref{energy1} we have: $E_{b_2}(\th_{b_1})-E_{b_1}(\th_{b_1})=(b_1-b_2)\int_0^1(1-s)\sin\th_{b_1}\text{d}s$, which is negative since $\sin\th_{b_1}$ is positive.

\qed

\begin{prop}
 The set of minimizers forms a branch of solutions.
 \label{branch}
\end{prop}

\begin{proof*}

\normalfont We start by  showing  that, if $b \to b_0$, then $\th_b$ is a minimizing sequence for $E_{b_0}$.
For every $\psi\in H^1_{0,1}$ one has:

$$
E_b(\psi)-E_{b_0}(\psi)=(b-b_0)\int_0^1(1-s)\sin\psi\text{d}s\le|(b-b_0)|\int_0^1(1-s)\text{d}s=\frac{1}{2}|b_0-b|.
$$
Therefore, recalling that $\th_b$ minimizes $E_b$, it follows:
\bea
E_{b_0} (\psi)&= E_{b} (\psi)+E_{b_0} (\psi)-E_{b} (\psi) \geq E_{b} (\th_b) -\frac 12 |b_0-b|  \\ 
& = E_{b} (\th_b)-E_{b_0} (\th_b)+ E_{b_0} (\th_b)-\frac 12 |b_0-b|  \\ \nn
&\geq E_{b_0} (\th_b)-  |b_0-b|. \nn
\eea

\noindent As the result holds for every $\psi\in H^1_{0,1}$, this implies that
$$
 E_{b_0}(\th_b) \leq E_{b_0}(\th_{b_0}) +  |b_0-b|. 
$$

\noindent This immediately implies that $\th_b$ is a minimizing sequence for $E_{b_0}$.  

\noindent We recall that $\th_{b_0}$ is a stationary point for $E_{b_0}$, and therefore the first Fr\'{e}chet derivative of $E_{b_0}(\theta(\cdot))$ vanishes in $\th=\th_{b_0}$. This implies that, recalling the second variation of $E(\th(\cdot))$ given in formula \eqref{secondvar}, we can write the difference $E_{b_0}(\th_{b})-E_{b_0}(\th_{b_0})$ as:
$$
E_{b_0}(\th_{b})-E_{b_0}(\th_{b_0})=\frac 12 \| (\th_b -\th_{b_0})' \|_{L^2_{0,1}} ^2 +\frac 12  b_0 \int_0^1 ds (1-s) \sin \xi (s) (\th_b - \th_{b_0})^2,
$$
where $\xi (s) = \alpha\th_{b}(s)+(1-\alpha)\th_{b_0} (s)$ for some $\alpha\in[0,1]$. Since $\sin \xi \geq 0$ we obtain
\be
\| (\th_b -\th_{b_0})' \|^2_{L^2_{0,1}} \leq E_{b_0}(\th_{b})-E_{b_0} (\th_{b_0})  \to 0,
\ee
as $b \to b_0$, which completes the proof. 
\end{proof*} 
 
\quad \qed

\noindent Summarizing, we have found a branch of  solutions which are minimizers of the energy. We call this branch {\it primary}. 
The next proposition excludes that 
such a branch bifurcates for some $b_0$ with the occurrence of a new  branch of solutions of \eqref{1} or, in other words, the primary branch $ b \to \th_b^{p}$ cannot generate any other branch of solutions.
 
\begin{prop}
It does not exist $b_0$ in which a new branch of solutions arises from the primary branch.
\end{prop}
\begin{proof}

Let us consider a value $b=b_0$ such that at $b_0$ a new branch of solutions arises.

\noindent Let $\th_b$ be a branch of solutions and denote by
$\th_b^p$ the primary one. 

Then
computing the first and second variation of the energy we find,
for a given $H^1$ function $h$ satisfying the boundary conditions $h(0)=0$,
\be
\label{secondvar2}
\frac{d^2} {d\l ^2} E_b( \th_b +\la h)\big|_{\la=0}= \int \ |h'|^2 +b \int  (1-s)   h^2\, \sin \th_b  >0.
\ee

\noindent The previous inequality holds because, in the hypothesis that at $b_0$ a new branch arises, one has that $\sin \th_b \approx  \sin \th^p_b \geq 0$ in the $H^1_{0,1}$ norm for $|b-b_0|$ small enough and for every solution $\th$ of equation \eqref{1} verifying the boundary conditions $\th(0)=0$ and $\th'(1)=0$.

This means that the new branch of solutions produced by the 
primary branch has to consist of local minima of the energy, at least for $|b-b_0|$ small enough.
But this is not possible. Indeed, let $\th_b^p$ be a minimizer and $\th_b$ a local minimizer. Let us consider the function
$$
f(\lambda):=\l \to E_b ( \l  \th_b^p + (1-\l) \th_b ).
$$
with $\l \in [0,1]$. Clearly $f(\lambda)$ has two local minima for $\l=0$ and $\l=1$ (with $f(1)<f(0)$), and therefore it must have a maximum for a certain $\lambda^*\in(0,1)$. If $|b-b_0|$ is suitably small then the function $\th^*:=\l^*\th_b^p+(1-\l^*)\th_b$ is such that $\|\th^*-\th_b^p\|_{H^1_{0,1}}$ is arbitrarily small, but in this case the sign of the second variation \eqref{secondvar2} tells us that all the directional (Gateaux) derivatives of $E$ are positive in $\th^*$, and therefore $\l^*$ cannot be a maximum for $f(\l)$.
\end{proof}

\noindent In order to parametrize the possible equilibrium configurations when $b$ varies in $\mathbb{R}^+$, let us consider the initial value problem
\eqref{1} with initial conditions $\th(0)=0$ and $\th'(0)=K$. Let $\th(s,K,b)$ be the corresponding unique solution. As we want to determine the solutions of the boundary condition problem expressed by \eqref{1}, \eqref{bc} and \eqref{bcbis}, we will plot (see Fig. \ref{bifurcation}
), in the plane $K,b$, the computed curves 
\begin{equation}
F(K,b):=\th'(1,K,b)=0.
\label{F_}
\end{equation} 
\begin{figure}[H]
\centering
\includegraphics[width=0.7\textwidth]{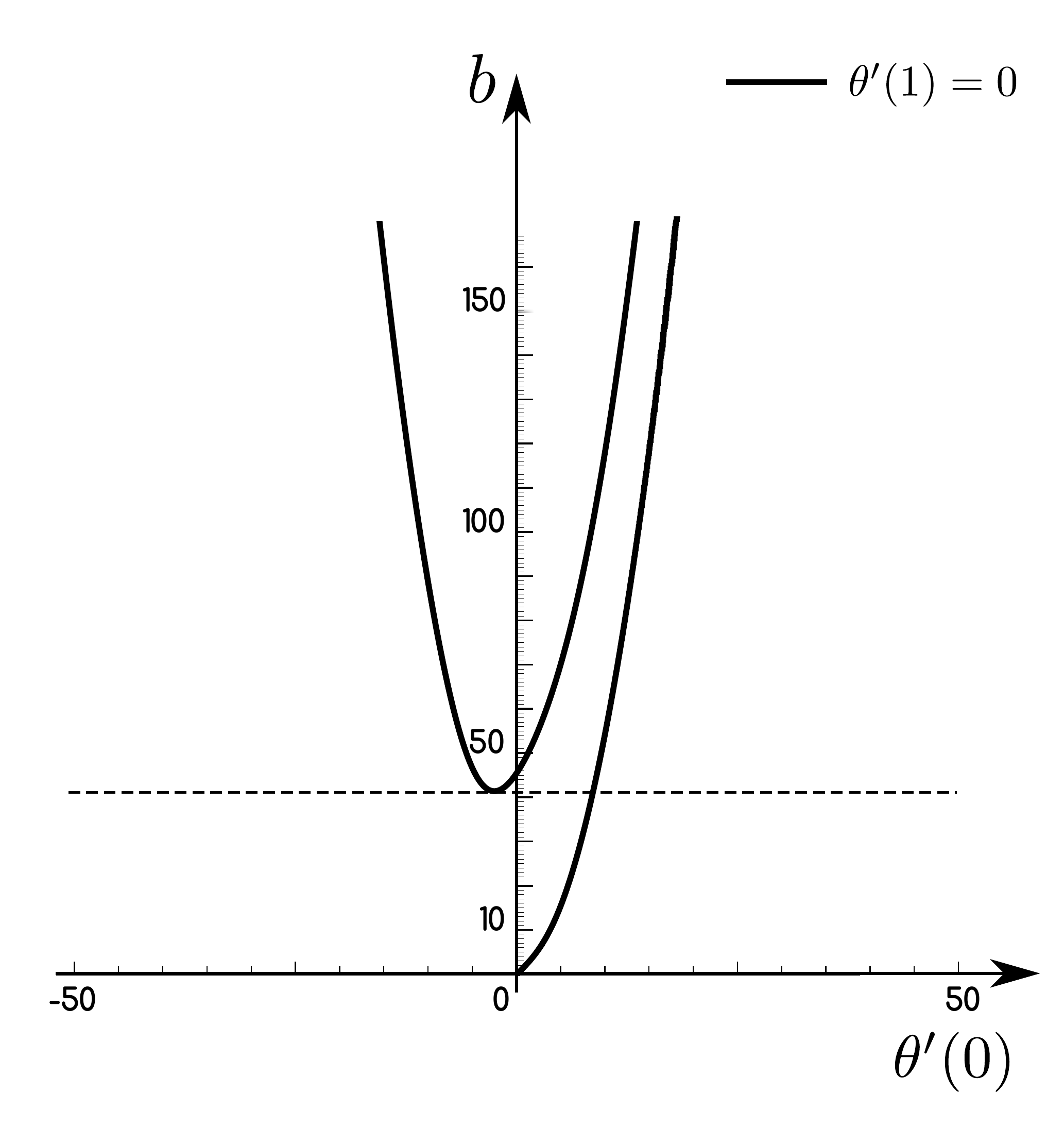} 
\caption{With a continuous line are represented points in the plane $(\th'(0),b)$ corresponding to solutions of the boundary value problem \eqref{1},\eqref{bc} and \eqref{bcbis}. The primary branch is on the right, while a secondary branch formed (as we will see) by both stable and unstable equilibrium shapes is also visible. Numerical evidence allows an estimate for the value $b_0\approx 41$ at which the secondary branch appears; the corresponding initial value of the curvature is $\approx -2.6$ (for details on the numerical procedure, see Section 6).}
\label{bifurcation}
\end{figure}
\noindent This is a preliminary step towards the establishment of a stability chart.

\noindent We remark that the level set $F(K,b)=0$ identifies a one-dimensional manifold (not necessarily connected) in the plane $(K,b)$, representing all the solutions 
of the boundary value problem  \eqref{1}, \eqref{bc}, \eqref{bcbis}, which includes all (possibly local) minima of the energy functional.  We have detected a first curve of solutions, the primary branch, constituted by all global minima of the energy.

\vskip .3cm
\begin{prop}
$K_b=\th_b'(0)$ diverges with $b$.
\end{prop}
\noindent \textit{Proof}
 
Let us  integrate the two members of \eqref{1} to get:
\be 
\label{Kb}
\th_b'(1)-\th_b'(0)=-b\int_0^1(1-s)\cos\th_b \text{d}s. 
\ee
Recalling \eqref{bc} and applying Chebyshev integral inequality (given that $1-s$ ad $\cos\th_b$ are both decreasing in $s$) we obtain:

\be
\label{Kb2}
K_b=\th_b'(0)\ge b\int_0^1(1-s)\text{d}s\int_0^1\cos\th_b \text{d}s=\frac{b}{2}\int_0^1\cos\th_b \text{d}s.
\ee
The last member can be written as 

\be 
\frac{b}{2}\int_0^1\cos\th_b \frac{\text{d}s}{\text{d}\th_b}\text{d}\th_b=\frac{b}{2}\int_0^{\th_b(1)}\frac{\cos\th_b}{\th_b'}\text{d}\th_b\ge \frac{b}{2K_b}\int_0^{\th_b(1)}\cos\th_b\text{d}\th_b=\frac{b}{2K_b}\sin\th_b(1),
\ee
the last inequality holding because $\th_b'$ is non increasing (by equation \eqref{1}). By \eqref{Kb} and \eqref{Kb2} we get $\left(K_b\right)^2 \ge\frac{b}{2}\sin\th_b(1)$. To prove that $K_b$ is unbounded when $b$ diverges we only have to exclude that $\th_b(1)\rightarrow 0$ when $b\rightarrow\infty$. Let us prove this by absurd. 

\noindent Suppose thus that $b_n$ is a sequence of positive reals such that $\lim_{n\rightarrow\infty}b_n=+\infty$, and let $\lim_{n\rightarrow\infty}\th_{b_n}(1)=0$. We recall that $\th_{b_n}(s)$ is positive and strictly increasing in $s$. This implies that $\th_{b_n}(s)\rightarrow 0$ pointwise for every $s$ in $[0,1]$. Since $\th_{b_n}$ is monotone and smooth for every $n$, this implies that the convergence to the limit function $\th(s)\equiv 0$ holds in the $H^1_{0,1}$ norm. Therefore, it should be $E_{b_n}(\th_{b_n})\rightarrow 0$ if $b_n\rightarrow + \infty$. But, since \eqref{bsmall} immediately implies that for small $b$ the energy is negative, the previous Lemma \ref{energymonotone} excludes this. 

\qed
\subsection{Further properties of the primary branch.}
\noindent The next two results will concern the behavior of the system when $b$ varies. Specifically, we will compute the derivative $\displaystyle{\frac{\partial E(\th)}{\partial b}}$ of the energy evaluated in the minimizer, and will study the derivative $\displaystyle{\frac{\partial\bar\th}{\partial b}}$ of the minimizer itself with respect to $b$.

\begin{prop}
Let $\th$ be the minimizer of the energy \eqref{energy1}. The derivative $\displaystyle{\frac{\partial E(\th)}{\partial b}}$ is given by $\int_0^1(s-1)\sin\th(s)\,\,ds$.
\end{prop}

\noindent \textit{Proof}
We write the derivative respect to $b$ of the deformation energy as follows:
\begin{eqnarray}\label{eqn:deform}
\frac{\partial}{\partial b}E_{DEF}=\frac{\partial}{\partial b}\int_0^1\frac{\th'^2(s)}{2}\,\,ds=\int_0^1\th'(s)\,\frac{\partial\th'(s)}{\partial b}\,\,ds\nonumber\\
=\int_0^1\left[\left(\th'(s)\,\frac{\partial\th(s)}{\partial b}\right)'-\th''(s)\,\frac{\partial\th(s)}{\partial b}\right]\,\,ds\nonumber\\
=\int_0^1b(1-s)\cos\th(s)\,\frac{\partial\th(s)}{\partial b}\,\,ds\nonumber\\
=\int_0^1b(1-s)\frac{\partial}{\partial b}\left(\sin\th(s)\right)\,\,ds
\end{eqnarray}
\noindent For the potential energy we have instead:
\begin{eqnarray}\label{eqn:poten}
\frac{\partial}{\partial b}E_{POT}=\frac{\partial}{\partial b}\left(\int_0^1-b(1-s)\sin\th(s)\,\,ds\right)\nonumber\\
=\int_0^1\left[-(1-s)\sin\th(s)-b(1-s)\frac{\partial}{\partial b}\sin\th(s)\right]\,\,ds
\end{eqnarray}
Since $E=E_{DEF}+E_{POT}$, equations (\ref{eqn:deform}) and (\ref{eqn:poten}) imply
\begin{eqnarray}
\frac{\partial}{\partial b}E_b=-\int_0^1(1-s)\sin\th(s)\,\,ds
<0
\end{eqnarray}
\noindent The last inequality follows from the fact that $0\le\th<\frac{\pi}{2}$ by Proposition \ref{prop1}, and confirms Lemma \ref{energymonotone}.

\qed

\begin{prop}
Let $\th(b,s)$ be the minimizer of the energy \eqref{energy1}. Then $\displaystyle{\frac{\partial\th(b,s)}{\partial b}}>0\quad\forall s \in (0,1]$. 
\end{prop}
\noindent \textit{Proof}
Let be $b_2>b_1$ and let $\th_{b_i}$ be the minimizers of $E_{b_i}$ ($i=1,\,2$). By absurd we suppose that
\begin{equation}
\exists \bar{s}\in\left(0,1\right]:\,\,\theta_{b_2}(\bar{s})<\theta_{b_1}(\bar{s})
\end{equation}
Since $\theta_{b_i}\in C^{\infty}[0,1]$, it is not empty the set $S$ of the intervals $I(\bar{s})$ containing $\bar{s}$, included in $[0,1]$ and such that $\,\,\theta_{b_2}(s)<\theta_{b_1}(s)\quad\forall s\in I(\bar{s})$.
Clearly there will exist an element of $S$ which is the largest (with respect to the inclusion relation); let us call $(s_1,s_2)$ this element. It will be $s_1<\bar{s}<s_2$ and, by continuity, $\theta_{b_1}(s_1)=\theta_{b_2}(s_1)$
and if $s_2<1$ then $\th_{b_2}(s_2)=\th_{b_1}(s_2)$ (we remark that it can be that $s_1=0$ and $s_2=1$). 
Let us define $\tilde{\theta}(s)$ by:
\begin{eqnarray}\tilde{\theta}:\left\{
\begin{array}{lr}
\theta_{b_2}(s) & \forall s\in[0,1]\backslash(s_1,s_2)\\
\theta_{b_1}(s) & \forall s\in(s_1,s_2)
\end{array}\right.
\end{eqnarray}
It is immediate to see that $\tilde{\theta}\in C^0[0,1]$ and $\tilde{\theta}\in C^{\infty}$ piece-wise, so that $\tilde{\theta}\in H^1_{0,1}$.

\noindent Let us evaluate now the difference
\begin{eqnarray}\label{eqn:thetaprimi}
E_{b_2}\left(\theta_{b_2}\right)-E_{b_2}\left(\tilde{\theta}\right)=\frac{1}{2}\int_{s_1}^{s_2}\left[\left(\theta'_{b_2}\right)^2-\left(\theta'_{b_1}\right)^2\right]\,\,ds\nonumber\\ -\int_{s_1}^{s_2}b_2(1-s)\sin\theta_{b_2}\,\,ds+\int_{s_1}^{s_2}b_2(1-s)\sin\theta_{b_1}\,\,ds
\end{eqnarray}
The sum of the last two terms is positive since $\theta_{b_1}>\theta_{b_2}$ and $0\le\th_{b_i}<\frac{\pi}{2}$ imply that $\sin\theta_{b_1}>\sin\theta_{b_2}$. Now, recalling equation \eqref{1} and boundary conditions \eqref{bc}, we can write $\theta'(s)$ as follows:
\begin{equation}\label{eqn:thprim}
\theta'(s)=-\theta'(1)+\theta'(s)=-\int_s^1\theta''(\sigma)\,\,d\sigma=\int_s^1b(1-\sigma)\cos\theta\,\,d\sigma
\end{equation}
We can therefore write the first term of the right hand side of \eqref{eqn:thetaprimi} as:
\begin{equation}
\label{inner}
\frac{1}{2}\int_{s_1}^{s_2}\left[\int_s^1b_2(1-\sigma)\cos\theta_{b_2}\,\,d\sigma\right]^2ds-\frac{1}{2}\int_{s_1}^{s_2}\left[\int_s^1b_1(1-\sigma)\cos\theta_{b_1}\,\,d\sigma\right]^2ds
\end{equation}
By hypothesis $b_2>b_1$. Moreover, since $0\le\th<\frac{\pi}{2}$, we have that $\theta_{b_2}<\theta_{b_1}\,\text{in}\,(s_1,s_2)\Rightarrow\cos\theta_{b_2}>\cos\theta_{b_1}$. Therefore we have $$\int_s^1 b_2(1-\sigma)\cos\theta_{b_2}\text{d}\sigma>\int_s^1 b_1(1-\sigma)\cos\theta_{b_1}\text{d}\sigma$$ 
Since the two integrals are positive, the previous inequality holds true for their squares. Therefore, the quantity in formula \eqref{inner} is positive, and so is the difference $E_{b_2}\left[\theta_{b_2}\right]-E_{b_2}[\tilde{\theta}]$, which is absurd since by hypothesis $\th_{b_2}$ is the minimizer of $E_{b_2}$.

\bigskip

\section{Other branches}

In Proposition \ref{thstab} below we prove that, for each fixed $b$ sufficiently large, there exist stationary points of the energy functional \eqref{energy1} for which $\th$ admits negative values in contrast with the solutions of the primary branch. These new solutions correspond to  local minimizers of the energy, thus representing  new stable solutions for Elastica. Referring to Figure 6, they correspond to suitable points on the represented new branch. For topological reasons, the set of stable solutions we detect via variational arguments, cannot stop. Thus other {stationary configurations} are necessarily present. Indeed, since the function $F(K,b)=\th'(1,K,b)$ is real analytic, the level set $\{(K,b)\in \mathbb{R}\times(0,+\infty)\,|\, F(K,b)=0\}$, which contains the local minimizers, cannot have extreme points (see e.g. \cite{sull}, Corollary 2, Example 1), hence it is unbounded. The stability character of these {stationary configurations} in this branch is for now unknown. {However we will establish in Section 5 some sufficient conditions (depending on $b$) for the stability of particular classes of solutions. In particular, we will show that if a solution has certain properties, then it is necessarily unstable.} Numerical evidence (see Figures \ref{Stable}, 16 in Section 6) shows the existence of these unstable solutions.

\vskip .2 cm
\noindent In order to state the main proposition of this section, let us consider the open convex cone ${\cal C}$:
 $$
 {\cal C}= \{ \th | \th \in H^1_{0,1}, \, \th(s) <0 \text{ for any } s\in(0,1] \}
 $$
 and look for the minima of $E$ restricted to such a cone. Let us define
 $$
  e_b({\cal C}):=\inf_{ \xi \in \CC } E(\xi)>-\infty
 $$
 
\noindent the last inequality obviously holding because of the Remark \ref{remark0}. The functional $E$ admits a minimum in $\overline{{\cal C}}$, the closure of $\mathcal{C}$ (see for instance \cite{Kudrila}, Theorem 7.3.8), i.e. there exists $\bar\th\in\overline{{\cal C}}$ such that $E(\bar{\th})=  e_b({\cal C})$. 

Since any element $\tilde\th$ of the primary branch is  positive, then $\tilde \th\notin \overline{{\cal C}}$ and hence $\bar \th\ne\tilde \th$. Moreover, for small $b$ we have an unique solution to the Euler-Lagrange equations, and it is positive. Then it follows that, for $b$ small$, \bar \th $ must be in $\partial\mathcal{C}$.
\begin{prop}
\label{thstab}
There exists $b_0 >0$, sufficiently large, such that for any $b>b_0$ there exists a local minimum of the energy  $\th $  which is solution of Eq. \eqref{1}, with $\th(0)=0$ and $\th'(1)=0$. Moreover $\th$ is negative and decreasing in $(0,1]$.

\end{prop}

\begin{proof}
\normalfont The main effort of the proof is in showing that $\bar \th \notin \partial\mathcal{C}$. Before that we cannot use equation \eqref{1} because we cannot consider variations.

 The basic starting point is that  $E(\bar{\th})<0$, provided that $b$ is sufficiently large. We begin  by proving this claim. 
 
 Consider the function $\xi\in \cal C$ defined by:
\begin{equation}
\begin{cases}
\displaystyle \xi(s) = -3\pi \frac s {2R} , \quad  s\in (0,R)\\
\displaystyle  \xi(s)= -3/2 \pi  \quad s \in (R ,1)
\end{cases}
\label{Folding}
\end{equation}

\noindent The total energy can be computed explicitly
$$
E(\xi)= \frac 94 \frac {\pi^2 }{R}+ b \int _0^R (1-s)   \sin (3\pi \frac s {2R} ) ds -b \int _R^1 (1-s) ds.
$$
Now we can choose  $R$ so small that the third term is dominating over the second one. Fixed $R$
we choose  $b$ so large that $E$ is negative.

\noindent In physical terms the above profile corresponds to a configuration of the beam folding around the origin for $3/4$ of the whole circle (of radius $R$), and then continuing vertically up to the end (see Fig. \ref{Folding_}).  The corresponding function $\xi(s)$ above defined is certainly not an equilibrium point for $E$, but its existence implies that it has to be  $E(\bar{\th})<0$.
\begin{figure}[H]
\centering
\includegraphics[width=0.4\textwidth]{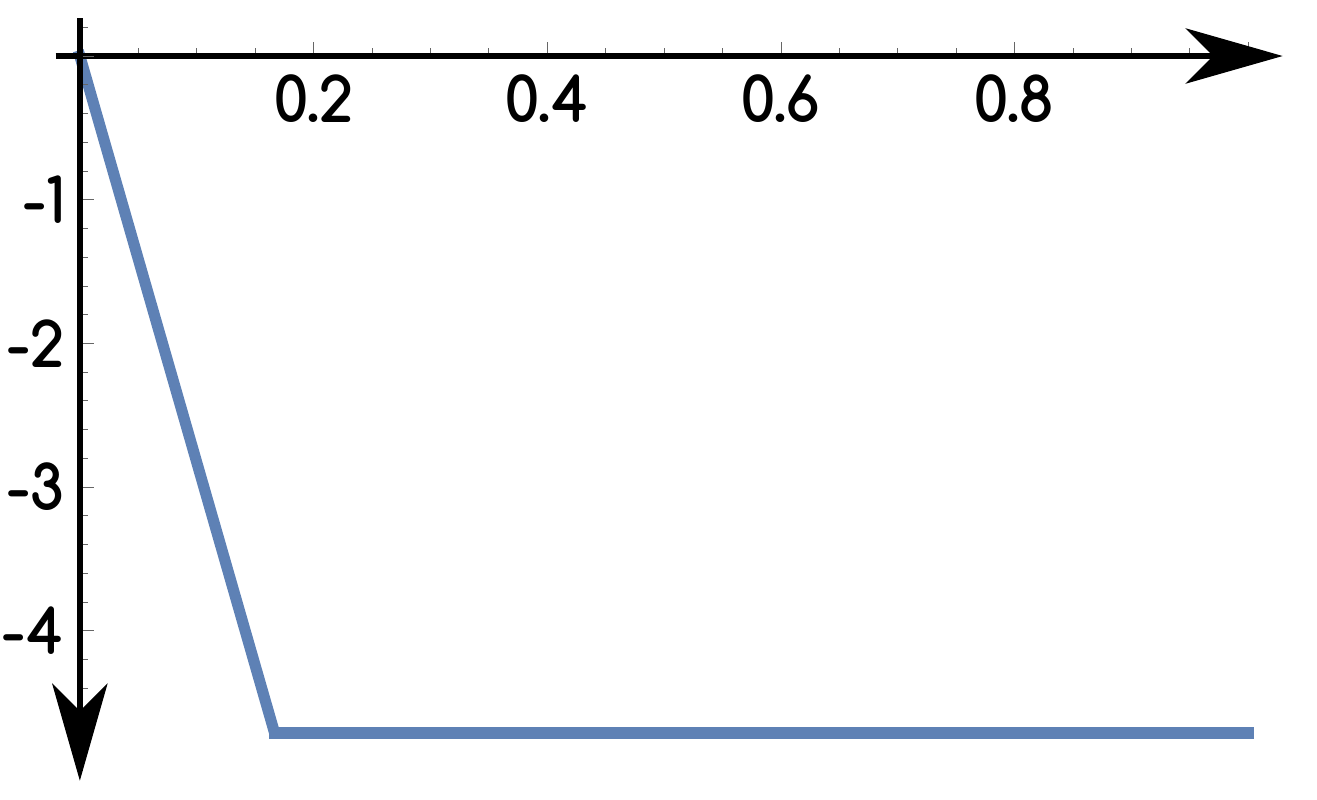}
\qquad
\qquad
\qquad
\includegraphics[width=0.25\textwidth]{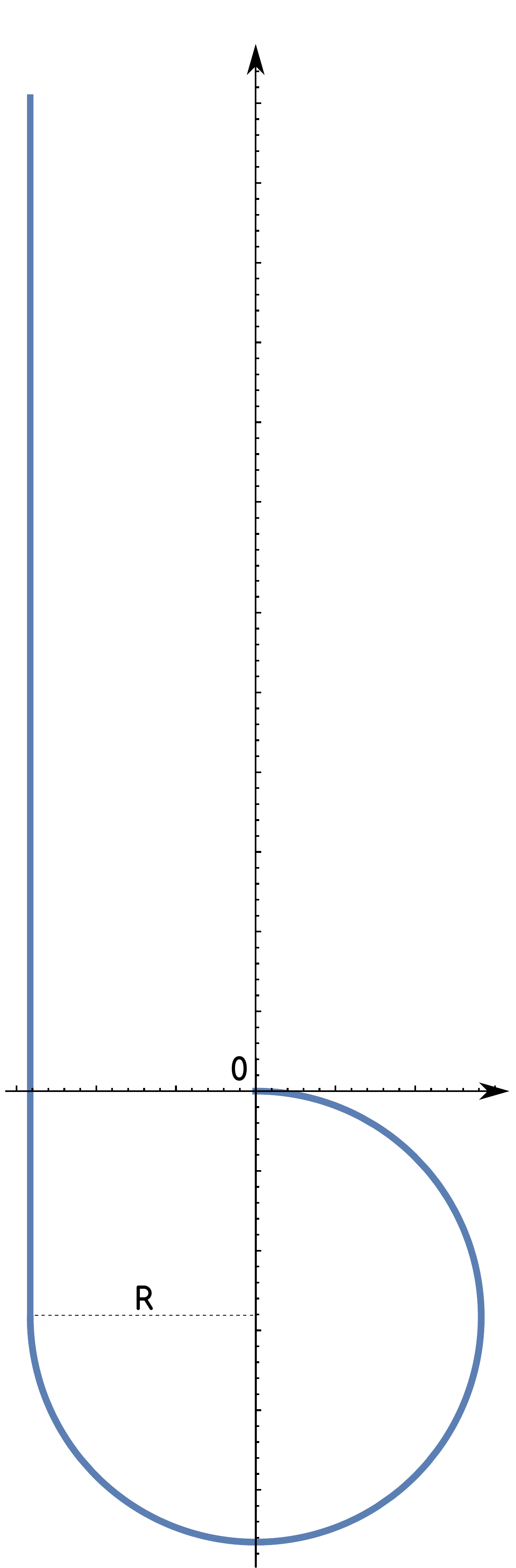} 
\caption{The function $\xi(s)$ defined in \eqref{Folding} (left) and the corresponding deformed shape of the \textit{Elastica} (right).} 
\label{Folding_}
\end{figure}
\noindent As second step we show that  $\bar{\th}$  cannot cross the axis $\th=- 3/2 \pi$. 

\noindent In facts, in this case, the shape of the \textit{Elastica} described by the function
$\xi=-3/2 \pi$ inside the set  $\{s:\bar{\th}(s)<-3/2 \pi \}$ and $\xi=\bar{\th}(s)$ outside that set, would be energetically more convenient because both the elastic and the potential energies are reduced.

\noindent As third step we prove that  $\bar{\th}$ must cross the axis $\th=-\pi$.  Otherwise, since $\bar{\th}\in\overline{\mathcal{C}}$, by hypothesis it is $\bar{\th}\le 0$, and therefore $\sin \bar{\th} (s) \leq 0$ for all $s \in (0,1)$. However this is not possible because, in this case,  the energy could not be negative.

\noindent Let $s_0$ be the first value of $s$ such that $\bar\th(s)=-\pi$. We show that $\bar\th$ is 
non increasing for $s>s_0$. 
Suppose the contrary. Then $\bar\th$ has local minimizers for $s>s_0$. We select some of them, $\bar s_1,\dots \bar s_n$ (and denote  $m_i=\bar\th(\bar s_i)$, $i=1,\dots,n$) as follows: let $\bar s_1$ be the first minimizer larger than $s_0$,  $\bar s_2$ is the first minimizer (if any) larger than $\bar s_1$ such that $m_2< m_1$, \dots $\bar s_{i+1}$ is the first minimizer larger than $\bar s_{i}$ such that $m_{i+1}< m_i$, for $i=1, \dots,n-1$. Let us now define the function
$$\xi(s)=\begin{cases}\bar\th(s)\quad \text{ for } s\notin\bigcup_{i=1}^n(\bar s_i,\bar s_{i+1}), \quad \bar s_{n+1}:= 1,\\
\displaystyle{\min_{s\in(\bar s_i,\bar s_{i+1})}}\{m_i, \bar\th(s)\}, \quad \text{ for } s\in (\bar s_i,\bar s_{i+1}).
\end{cases}
$$
\noindent The construction is illustrated in Figure 8.
Clearly $\xi \in H^1$ since by definition $\xi(s_i)=\bar{\th}(s_i)$ for any $i$.
\begin{figure}[H]
\centering
\includegraphics[width=0.4\textwidth]{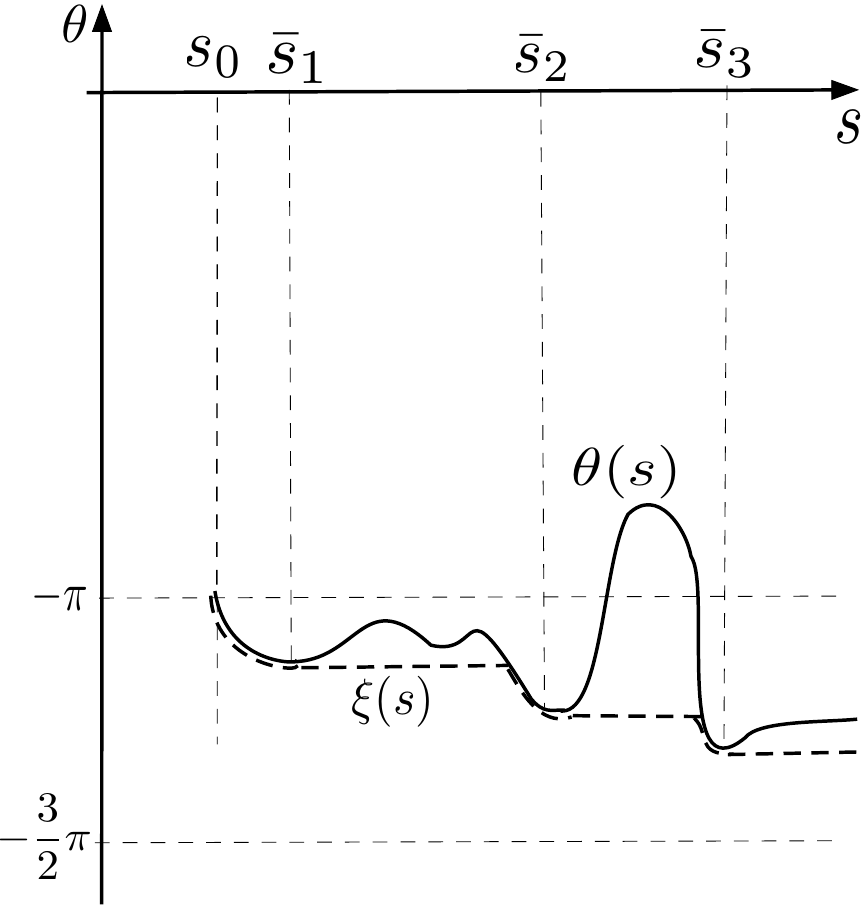}
\caption{$\theta$ continuous line, $\xi$ dashed line.}
\end{figure}
\noindent The configuration $\xi$ is energetically more convenient because the elastic energy is decreased in the constant intervals and the potential energy is decreased because $-\sin(\, \cdot\,)$ is positive above $-\pi$ and increasing in $(-\frac 3 2 \pi, -\pi)$.

\noindent Next we prove that $\bar \th$ is not increasing also in $(0,s_0)$. Suppose the contrary. Since $\bar{\th}\in\overline{\mathcal{C}}$, it is not increasing in a right neighborhood of zero, the set of local minimizers for $\bar{\th}$ in $(0, s_0)$ is not empty since $\bar{\th}(s_0)=-\pi$. Let $s_m$ be the first minimizer; the set of local maximizers is also not empty. Let $s_M$ be the first maximizer. It is obviously $s_m<s_M<s_0$. 

\noindent \underline{Case a)} Let us assume that $\th (s_m) >\pi/2$.

\noindent Since $\bar{\th}$ is monotonic in $[0,s_m]$, there exists a unique $\bar{s}$ such that $\bar{\th}(\bar{s})=\bar{\th}(s_M)$.

\noindent Then we consider the function $\xi(s)$ defined by 
$$\xi(s)=\bar{\th}(s)\quad\text{for}\quad s\in [0,\bar{s}]$$ 
$$\xi(s)=\bar{\th}(s_M)\quad\text{for}\quad s\in (\bar{s},s_M]$$ 
$$\xi(s)=\bar{\th}(s)\quad\text{for}\quad s\in (s_M,1]$$

\noindent The function $\xi$ would clearly be energetically more convenient than $\bar{\th}$. 

\noindent \underline{Case b)} Suppose instead that $\th (s_m) \leq \pi/2$

\noindent Since $\bar{\th}(s_0)=-\pi$, it is not empty the set $N$ of points such that $\bar{\th}(s)=\bar{\th}(s_m)$. Let $\bar{\bar{s}}$ be the minimum of $N$. Then we consider the function $\xi(s)$ defined by
$$\xi(s)=\bar{\th}(s)\quad\text{for}\quad s\in [0,s_m]$$ $$\xi(s)=\bar{\th}(s_m)\quad\text{for}\quad s\in (s_m,\bar{\bar{s}}]$$ $$\xi(s)=\bar{\th}(s)\quad\text{for}\quad s\in (\bar{\bar{s}},1]$$

\noindent Again the function $\xi$ would be energetically more convenient than $\bar{\th}$.

\noindent In conclusion, the function $\bar \th$ is non increasing in $(0,1)$. Hence the only possibility to be in $\partial \mathcal{C}$ is that there is $\lambda>0$ such that $\bar\th(s)=0$ for $s\in (0,\lambda)$. We now show that this is not possible.

\noindent Suppose that such a $\lambda$ exists. We define
\begin{equation}
\xi(s)=\bar{\th}(s+\l)  \quad \text {for} \quad s\in (0,1-\l)\\
\qquad \text{and } \quad \xi(s)=\bar{\th}(1) \quad \text {for} \quad s\in (1-\l,1)
\label{exclude}
\end{equation}
\noindent An explicit computation leads us to  (defining $B:=\sin \bar{\th}(1)$)
$$
E(\xi)=\frac 12 \int_0^1 |\xi'|^2 ds -b \int_0^{1-\l} (1-s) \sin (\bar{\th} (s+\l)) ds - b B \int_{1-\l}^1 (1-s) ds.
$$
\noindent Moreover, using that $ \int_0^1 |\xi'|^2 ds = \int_0^1 | \bar{\th}'|^2 ds$, by means of an obvious change of
variables
$$
E(\xi)-E(\bar{\th}) = -b \l \int _{\l}^1 \sin \bar{\th} (s) ds - b B \int_{1-\l}^1 (1-s) ds
$$

\noindent We show that the above quantity is negative. Indeed the last term is negative because $\bar\th(1)\in (-\frac 3 2 \pi, -\pi)$.
As for the first term of the right hand side, we write
\begin{equation}
T= \int _{\l}^1 \sin \bar{\th} (s) ds = \int _{\l}^{s*} \sin \bar{\th} (s) ds +\int _{s*} ^1 \sin \bar{\th} (s) ds 
\end{equation}
\noindent where $s*$ is a point for which $ \bar{\th}(s*)=-\pi$. We remark that, because of the monotonicity of $\bar{\th}$, it is $\l<s^*$. Suppose by absurd that $T<0$, namely
$$
- \int _{\l}^{s*} \sin \bar{\th} (s) ds > \int _{s*} ^1 \sin \bar{\th} (s) ds.
$$
\noindent Then it holds
\begin{align*}
 &-\int _{\l}^{s*} (1-s) \sin \bar{\th} (s) ds > -\int_{\l}^{s*} (1-s^*) \sin \bar{\th} (s) ds > \\ 
 &\int _{s*} ^1 (1-s^*)  \sin \bar{\th} (s) ds > \int _{s*} ^1 (1-s)  \sin \bar{\th} (s) ds.
\end{align*}
\noindent Therefore
$$
  -b \int _0^1 (1-s)  \sin \bar{\th} (s) ds > 0
 $$
\noindent but this contradicts the negativity of $E(\bar{\th})$.

\noindent Summarizing, we are left with the two possibilities depicted in Figs. 9, 10.
\begin{figure}[H]
\centering
\includegraphics[width=0.4\textwidth]{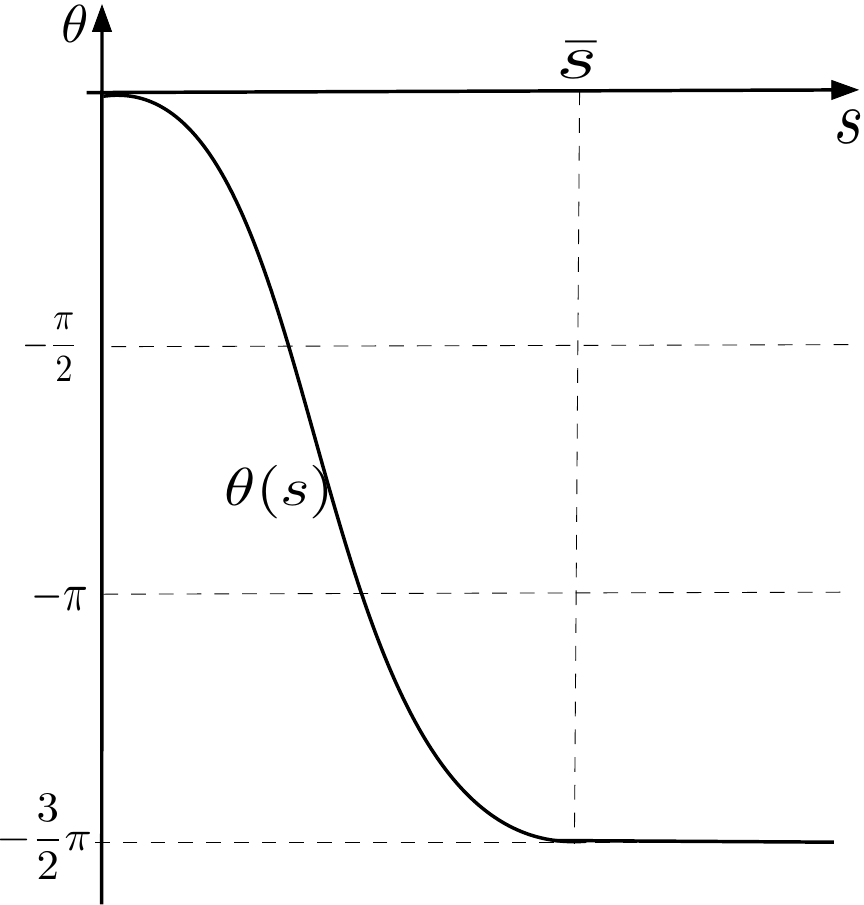}
\caption{$\theta=-\frac 3 2\pi$ for $s\ge \bar s$.}
\end{figure}
\begin{figure}[H]
\centering
\includegraphics[width=0.4\textwidth]{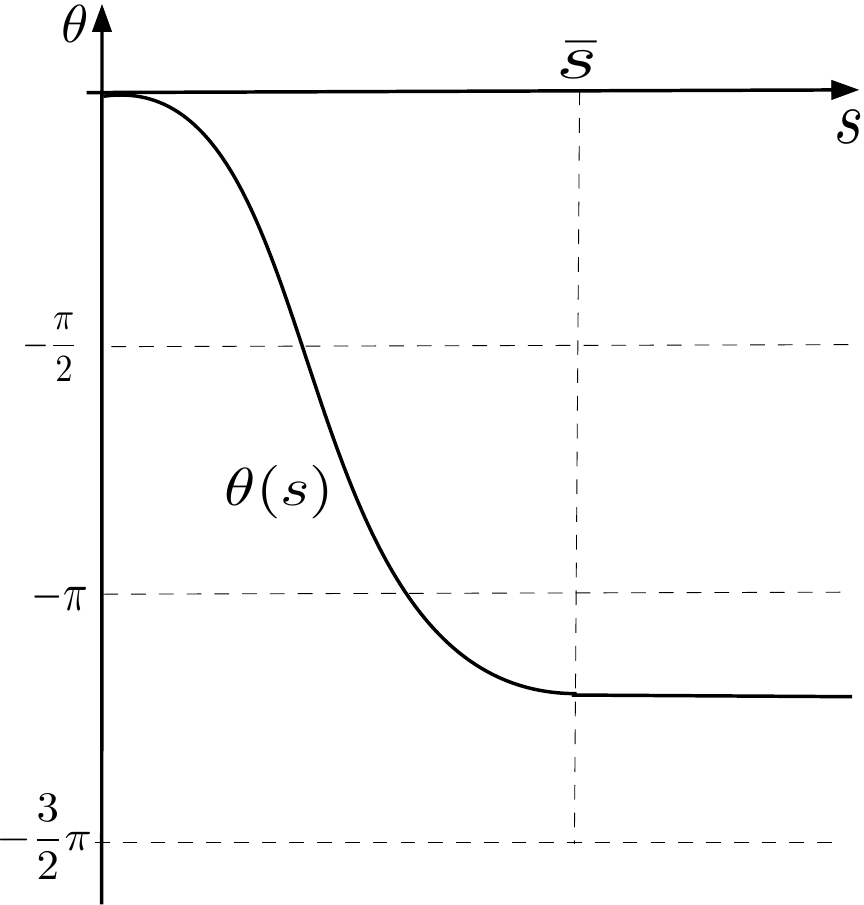}
\caption{The only possible $\theta$.}
\end{figure}

\noindent Note that in both  cases $\bar{\th}$ is decreasing so that it cannot vanish in $(0,1]$. This is enough to conclude that
$\bar{\th} \in \CC$.  Therefore we are now allowed to use Eq. \eqref {1}, \eqref{bc}. The same uniqueness argument, already used for the primary branch, 
excludes also the possibility that $\bar{\th}$ takes the constant value $-3/2 \pi$ over a finite interval.

\noindent In conclusion $\bar{\th}$ is strictly decreasing in $(0,1]$ and $\lim_{s \to 1} \bar{\th}'(s)=0$. The continuity $b \to \bar{\th}_b$ follows the same lines as shown in Proposition \ref{branch}. 

\end{proof}

\noindent We want now to further study the solution just discussed, and prove that it is obtained by ``gluing'' together the solutions of two different problems.
\noindent Let $\tilde{\theta}(s)$ be a solution of the boundary value problem: 

\begin{equation}
\theta''(s)=-b(L-s)\cos\theta(s)\quad\quad
\theta(0)=0\quad\quad
\theta'(L)=0
\label{BV}
\end{equation}
such that:

\begin{itemize}
\item[a.] $\tilde{\theta}(s)$ is decreasing;
\item[b.] its range is contained in $[0,-\frac{3}{2}\pi]$;
\item[c.] there exists a unique $\bar{s}$ such that $\tilde{\theta}(\bar{s})=-\pi$.
\end{itemize}

\noindent Note that these hypotheses are satisfied by the solution corresponding to the deformed shape represented in Fig.\ref{Stable}.   

\begin{prop}
Suppose that there exists a solution of the boundary value problem \ref{BV} satisfying requirements a, b and c given before. Then the relative deformed shape restricted  to $[\bar{s},L-\bar{s}]$ corresponds (up to a reflection) to the deformed shape of a beam of length $L-\bar{s}$ in the minimizer of the energy (primary solution).
\end{prop}

\begin{proof} Let us set $\phi(s)=-\tilde{\theta}(s)-\pi$. Obviously $\phi$ is increasing, $\phi''=-\tilde{\theta''}$ and $\cos{\tilde{\theta}}=-\cos\phi$. We have therefore that:

\begin{equation}
\phi''(s)=-b(L-s)\cos\phi(s)
\end{equation} 
holds with boundary conditions:

\begin{equation}
\phi(0)=-\pi \quad \quad \quad \phi'(L)=0
\end{equation} 
Let us also set $\xi=s-\bar{s}$ and $\gamma(\xi)=\phi(\bar{s}+\xi)$. We consider now the restriction of $\gamma(\xi)$ on the interval $I_{\bar{s}}=[0,L-\bar{s}]$. It is immediate to verify that $\gamma(\xi)$ solves the boundary value problem:

\begin{equation}
\gamma''(\xi)=-b(L-\bar{s}-\xi)\cos\gamma(\xi)\quad\quad
\gamma(0)=0\quad\quad
\gamma'(L-\bar{s})=0
\end{equation}
Moreover, $\gamma$ is increasing in $I_{\bar{s}}$ and its range is contained in $[0,\frac{\pi}{2}]$. Therefore, recalling Proposition \ref{prop1},  $\gamma(\xi)$ corresponds to the (unique) absolute minimum of the energy functional when the length of the \textit{Elastica} is equal to $L-\bar{s}$, which completes the proof.

\end{proof}

\noindent As an internal consistency check of the numerical tools employed, we plotted a superposition between the curled stable solution and the absolute minimum solution of suitable reduced length (and rotated) in Fig.~\ref{overlap}.
\begin{figure}[H]
\begin{center}
\includegraphics[scale=0.25]{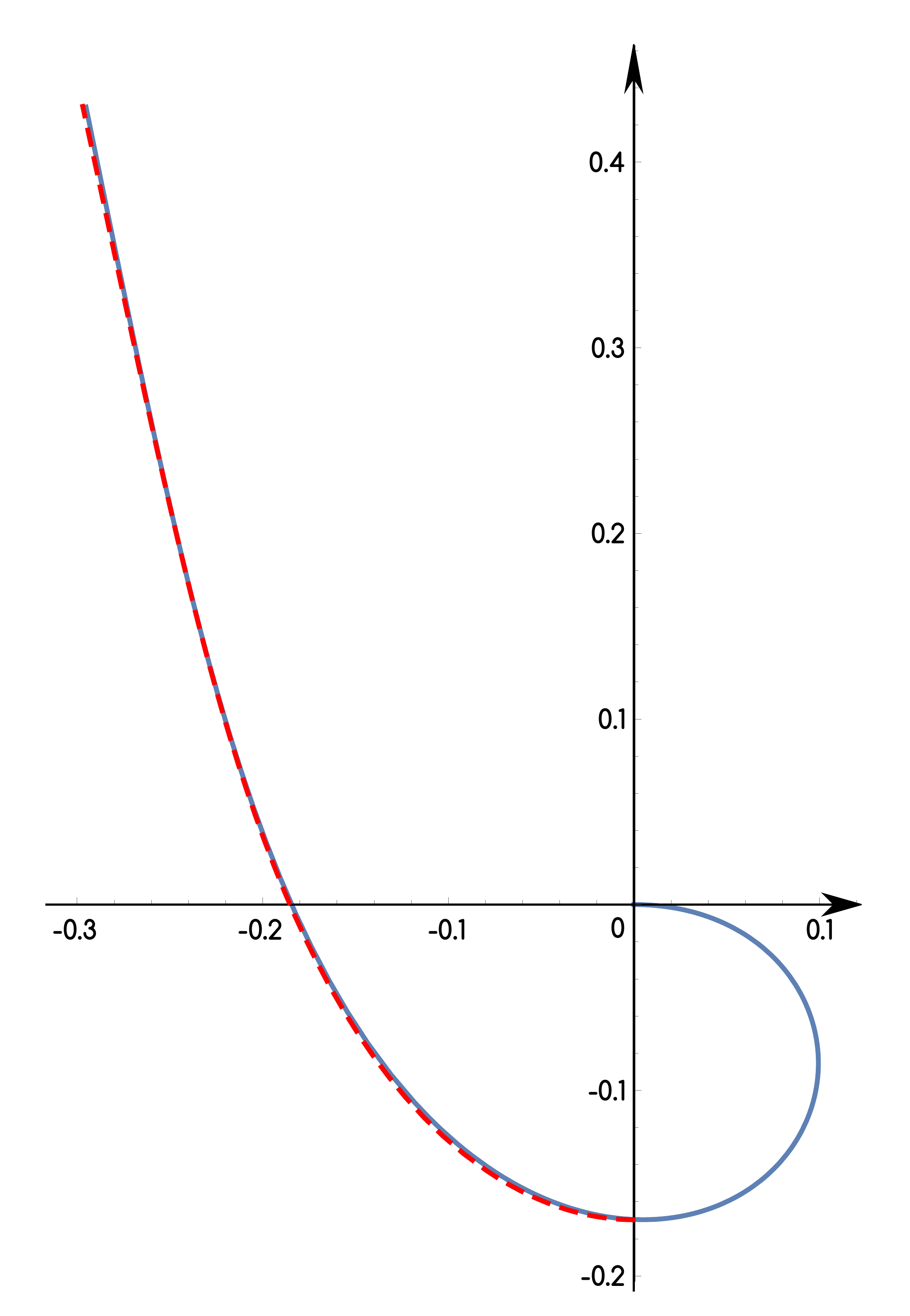}
\end{center}
\caption{Superposition of the curled stable solution (continuous line) and the absolute minimum solution (dashed line).}
\label{overlap}
\end{figure} 

\section*{5. Results on stability for two classes of solutions}
 
Let us recall the second variation of the energy functional $E$:

\begin{equation}
\mathcal{V}=\int_0^1[(h')^2+b(1-s)\sin\theta(s)h^2]\text{d}s
\label{1b}
\end{equation}

\noindent We start by proving the following result, which can be applicable for solutions $\theta$ (possibly different from the local/global minimizers analyzed in the previous sections) such that $\sin\theta$ is positive in a neighborhood of 1. 

\begin{prop}
\label{stability}
Let $\theta(s)$ be a stationary point for $E$ for a given $b$. If there exists $\lambda \in (0,1)$ such that $\sin\theta$ is non negative in $[\lambda,1]$, then $\mathcal V$ is positive if $b<\frac{\pi^2}{2\lambda^3(2-\lambda)}$
\end{prop}

\begin{proof*}
\normalfont Since $\mathcal{V}[h]=\mathcal{V}[|h|]$, we can limit ourselves to variations that are non negative.

\noindent We can suppose that $\sin\theta$ is negative in the whole interval $(0,\lambda)$, as if not $\mathcal{V}$ will be increased by replacing $\sin\theta$ with $-\sin\theta$ on the subsets of $(0,\lambda)$ in which it is negative.

\noindent The stability of the solution $\theta$ is then equivalent to

\begin{equation}
\int_0^1(h')^2ds+b\int_{\lambda}^1(1-s)\sin\theta(s)h^2ds>b\int_0^\lambda(1-s)|\sin\theta(s)|h^2ds
\label{2}
\end{equation}
which holds if so does

\begin{equation}
\int_0^{\lambda}(h')^2ds>b\int_0^\lambda(1-s)h^2ds
\label{3}
\end{equation}
\end{proof*}
\noindent Before proceeding, let us prove the following 

\begin{lem}
Without losing generality, we can establish \eqref{3} only for variations $h$ for which there exists a $C^0$ non decreasing representative. 
\end{lem}
 
\begin{proof*}
\normalfont Let us suppose that $h$ has a $C^0$ representative (which we for simplicity will also denote by $h$) that decreases in a sub-interval $(\alpha,\beta)$ of $(0,\lambda)$. We can then define $\bar{h}\in C^0$ as

\begin{equation}
\bar{h}:= 
\begin{cases}
h & \forall s\in (0,\alpha)\\
2h(\alpha)-h(s) & \forall s\in (\alpha,\beta)\\
h(s)+2[h(\alpha)-h(\beta)] & \forall s \in (\beta,\lambda)
\end{cases}
\end{equation}
It is easily seen that $\bar{h} \in H^1$ and that it is non negative and non decreasing in $(\alpha,\beta)$. This implies that, for $s\in(\alpha,\beta)$, 

\begin{equation}
\bar{h}^2(s)=4[h(\alpha)]^2+h^2(x)-4h(\alpha)h(s) \ge h^2(s)
\end{equation} 
and since obviously $\bar{h}'=h'$, one gets

\begin{equation}
\int_{\alpha}^{\beta}(\bar{h}')^2ds>b\int_{\alpha}^{\beta}(1-s)\bar{h}^2ds \Longrightarrow \int_{\alpha}^{\beta}(h')^2ds>b\int_{\alpha}^{\beta}(1-s)h^2ds
\end{equation} 
from which, observing that $\bar{h}\ge h$ in $(\beta,\lambda)$ and that piecewise monotonic functions are dense  in $H^1$, the thesis follows.

\qed

\noindent We can now complete the proof of Proposition \ref{stability}.
For $H^1$ functions $h$ that vanish in only one of the extrema of a real interval, a weak version of Poincar\'e inequality holds\footnote{\noindent This is easily seen applying the classical Poincar\'e inequality to $\tilde{h}(x):=h(x)$ for $x\in [0,\lambda]$ and $h(2\lambda-x)$ for $x \in (\lambda, 2\lambda]$.}, namely there exists $C$ such that

\begin{equation}
\|h\|_{{L^p}_{(0,\lambda)}} \le C\|h'\|_{{L^p}_{(0,\lambda)}}
\end{equation}
which, for $p=2$ provides

\begin{equation}
\left(\int_0^{\lambda}h^2ds\right)^{\frac{1}{2}}\le C\left(\int_0^{\lambda}(h')^2ds\right)^{\frac{1}{2}}
\end{equation}
where $\frac{2\lambda}{\pi}$ is the optimal value for the constant\footnote{The optimal Poincar\'e constant would be simply $\frac{\lambda}{\pi}$ for functions vanishing in both extrema of the interval.}. We have thus

\begin{equation}
\int_0^{\lambda}h^2ds\le \frac{4l^2}{\pi^2}\int_0^{\lambda}(h')^2ds
\end{equation}
Therefore, \eqref{3} holds if so does

\begin{equation}
\frac{\pi^2}{4{\lambda}^2}\int_0^{\lambda}h^2ds>b\int_0^\lambda(1-s)h^2ds
\label{5}
\end{equation}
Since $1-s$ is decreasing, and recalling that $h$ is non decreasing, by Chebyshev's integral inequality we have

\begin{equation}
\int_0^\lambda(1-s)h^2ds \le \int_0^\lambda(1-s)ds\int_0^\lambda h^2ds
\end{equation}
Therefore, \eqref{5} holds if

\begin{equation}
\frac{\pi^2}{4{\lambda}^2}\int_0^{\lambda}h^2ds>b\int_0^\lambda(1-s)ds\int_0^\lambda h^2ds
\end{equation}
Solving the previous inequality with respect to $b$, we get

\begin{equation}
b<\frac{\pi^2}{4{\lambda}^2}\left(\lambda-\frac{\lambda^2}{2}\right)^{-1}=\frac{\pi^2}{2\lambda^3(2-\lambda)}
\end{equation}
which completes the proof.
\end{proof*}

\qed

\noindent We want now to establish a check on instability. Let $\theta$ be a stationary point for $\mathcal{E}$. We will prove the following

\begin{prop}
\label{unstable}
Let us suppose that there exist $\lambda$, $\mu$, $\nu$ such that $0<\lambda<\mu<\nu<1$, and that $\sin\theta>0$ in $(0,\lambda)$ and $\sin\theta<0$ in $(\mu,\nu)$. Then $\theta$ is not a minimizer for $E$ if\footnote{It is actually easy to get rid of $\mu$ and obtain a similar result supposing that $\sin\theta$ simply changes its sign at $\nu$; however, we prefer the current formulation in order to be able to apply the inequality when the zero of $\sin\theta$ is known with an error (that can be controlled).} $$b>\frac{12}{\lambda(3\lambda^2-16\lambda+12)\int_\mu^\nu|\sin\theta| ds}$$    
\end{prop}

\begin{proof*}
\normalfont We consider a variation $\bar{h}(s)$ defined as:
\begin{equation}
\bar{h}:= 
\begin{cases}
\frac{M}{\lambda}s & \forall s\in (0,\lambda)\\
M & \forall s \in (\lambda,1)
\end{cases}
\end{equation}

\noindent Let us compute $\mathcal{V}$ for this variation\footnote{We will see that the inequality is not depending on the value $M$.}

\begin{equation}
\mathcal{V}[\bar{h}]=\frac{M^2}{\lambda}+b\left(\frac{M}{\lambda}\right)^2\int_0^\lambda s^2(1-s)\sin\theta ds+ bM^2\int_\lambda^1(1-s)\sin\theta ds
\end{equation}
Recalling the hypotheses on the sign of $\sin\theta$, $\mathcal{V}$ is negative if

\begin{equation}
\frac{M^2}{\lambda}+b\left(\frac{M}{\lambda}\right)^2\int_0^\lambda s^2(1-s) ds < bM^2\int_\mu^\nu(1-s)|\sin\theta| ds
\end{equation}
that is

\begin{equation}
\frac{M^2}{\lambda}+bM^2\left(\frac{\lambda}{3}-\frac{\lambda^2}{4}\right) < bM^2\int_\mu^\nu(1-s)|\sin\theta| ds
\end{equation}
which holds if so does

\begin{equation}
\frac{1}{\lambda}+b\left(\frac{\lambda}{3}-\frac{\lambda^2}{4}\right) < b(1-\lambda)\int_\mu^\nu|\sin\theta| ds
\end{equation}
Solving for $b$, one gets

\begin{equation}
b>\frac{12}{\lambda(3\lambda^2-16\lambda+12)\int_\mu^\nu|\sin\theta| ds}
\label{44}
\end{equation}
\end{proof*}

\qed

\noindent We remark that the right hand side of this inequality involves a function depending on $b$. In this case, contrarily to the previous one, in order to apply the check one needs not only to know the behavior of the sign of $\sin\theta$, but also an estimate of the integral of its negative part.

\section*{6. Some numerical results}

The previous results can be applied to check the stability character of a stationary point $\theta$, when they are not obtained as local/global minimizers. We will apply them in particular to a solution only found numerically, which will be described below.   

\noindent To solve numerically the boundary value problem \eqref{1}, \eqref{bc} and \eqref{bcbis} we used a standard shooting technique. Specifically, we first solved a set $S$ of initial value problems with suitably parametrized initial data, i.e. the problem given by \eqref{1}, \eqref{bcbis} and $\th'(0)=K$.
Next we searched for the solution of the boundary value problem among the solutions of $S$. The non linear second order Chauchy  problem has been solved using the built in function of \textit{Mathematica} NDSolve for the numerical solution of differential equations.
This function can use different methods for solving differential equations, but basically all of them employ an adaptive step size in order to stay within a certain error range. The method for the solution is automatically chosen taking into account different properties of the problem, but it is mainly based on the evaluation of its stiffness. NDSolve returns a  function interpolating the values evaluated at the  extrema of the steps. Parameterizing the Chauchy problem with respect to $\theta'(0)=K$, with $-40\le K \le 40$, we have determined the values $\theta'(1,K,b)$.  The error on these values depends of course on the step size chosen for $K$. {Specifically, by means of the option \textit{AccuracyGoal} of package NDSolve, we can control the desired final output up to an (adimensional) absolute error chosen \textit{a priori} (we selected $10^{-8}$); this was done for the estimate of all the quantities involved in the numerical check of the inequalities, i.e. $\th$, $\sin\th$ and $\int\sin\th$. }

\noindent With $b=60$, numerical simulations show that there is a solution $\bar{\theta}(s)$ such that $\sin\bar{\theta}$ is non negative in $[0.3,1]$, with an error (see above) sufficiently small to comfortably allow certainty on the sign at the left extremum of the interval. Since $\bar{\theta}$ appears in this case to be everywhere negative, the previous Proposition \ref{thstab} tells us that the solution can be a local minimizer. Indeed, since in this case $\frac{\pi^2}{2\lambda^3(2-\lambda)}\approx 107.5$, Proposition \ref{stability} implies then that this is a stable configuration of the \textit{Elastica} (not belonging to the primary branch).

\noindent The graph of $\sin\bar{\theta}$ is shown in figure \ref{sintheta1}.

\begin{figure}[H]
\centering
\includegraphics[scale=0.35]{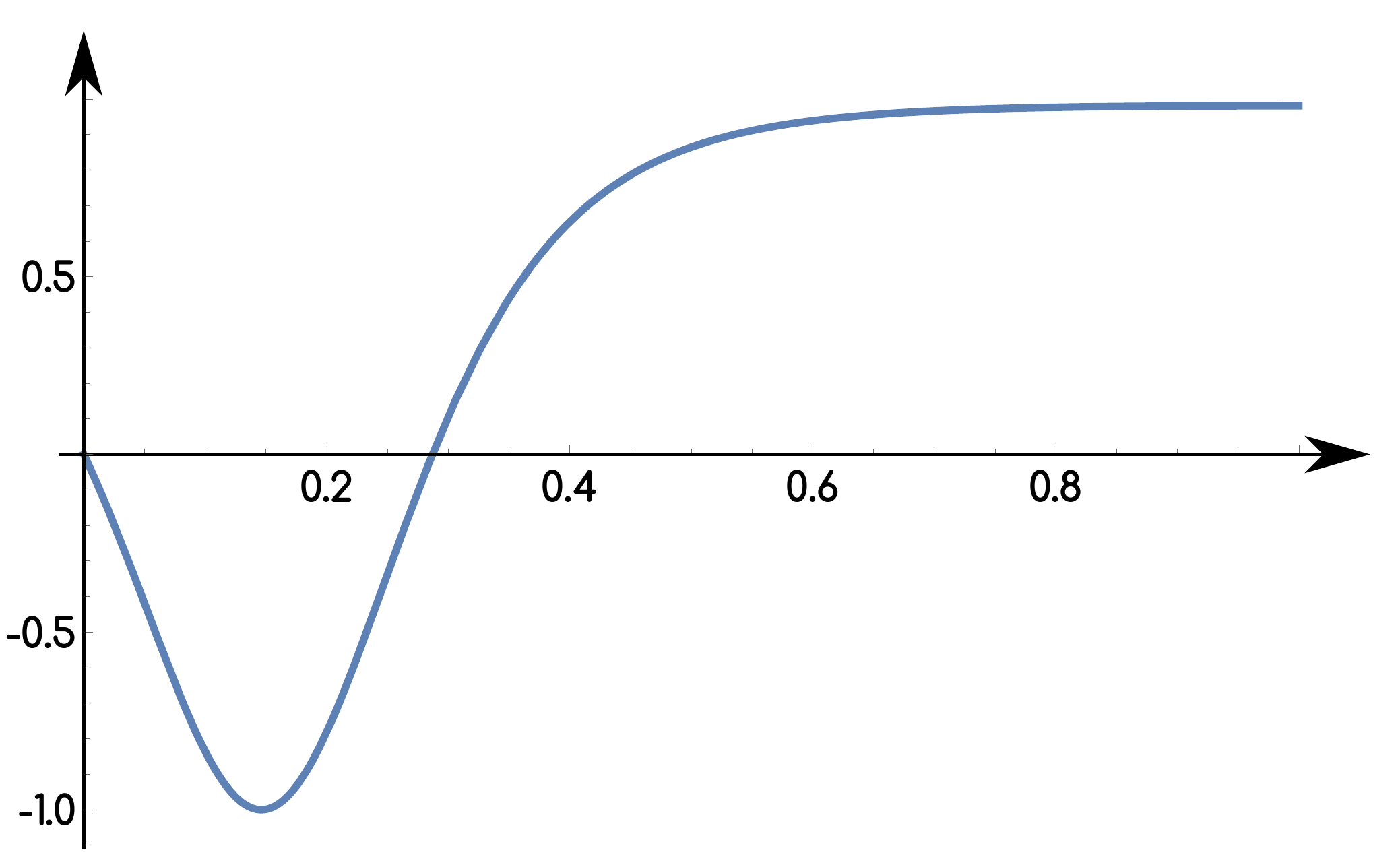}
\caption{Graph of $\sin\bar\theta$ (stationary point for $b=60$)}\label{sintheta1}
\end{figure}
 
\noindent With the same value $b=60$, another solution $\hat{\theta}(s)$ is numerically found, and simulations show that $\sin\hat{\theta}$ is positive in $(0,0.11)$ and negative in $(0.14,0.5)$ with an error (see above) sufficiently small to comfortably allow certainty on the sign at the extrema of the intervals (notice that in this case Proposition \ref{thstab} is not applicable because the solution is not everywhere negative). Numerical results show that $\int_{0.14}^{0.55}|\sin\hat{\theta}(s)|ds$ is not smaller than 0.2. Since in this case the right hand side of the inequality \eqref{44} is $\approx 53.25$, this means, by the previous Proposition \ref{unstable}, that this solution is not stable. 

\noindent The graph of $\sin\hat{\theta}$ and of $\int_0^s\sin\hat{\theta}ds$ are shown in Figs. \ref{ffffff} and \ref{gggggg}.

\begin{figure}[H]
\centering
\includegraphics[scale=0.5]{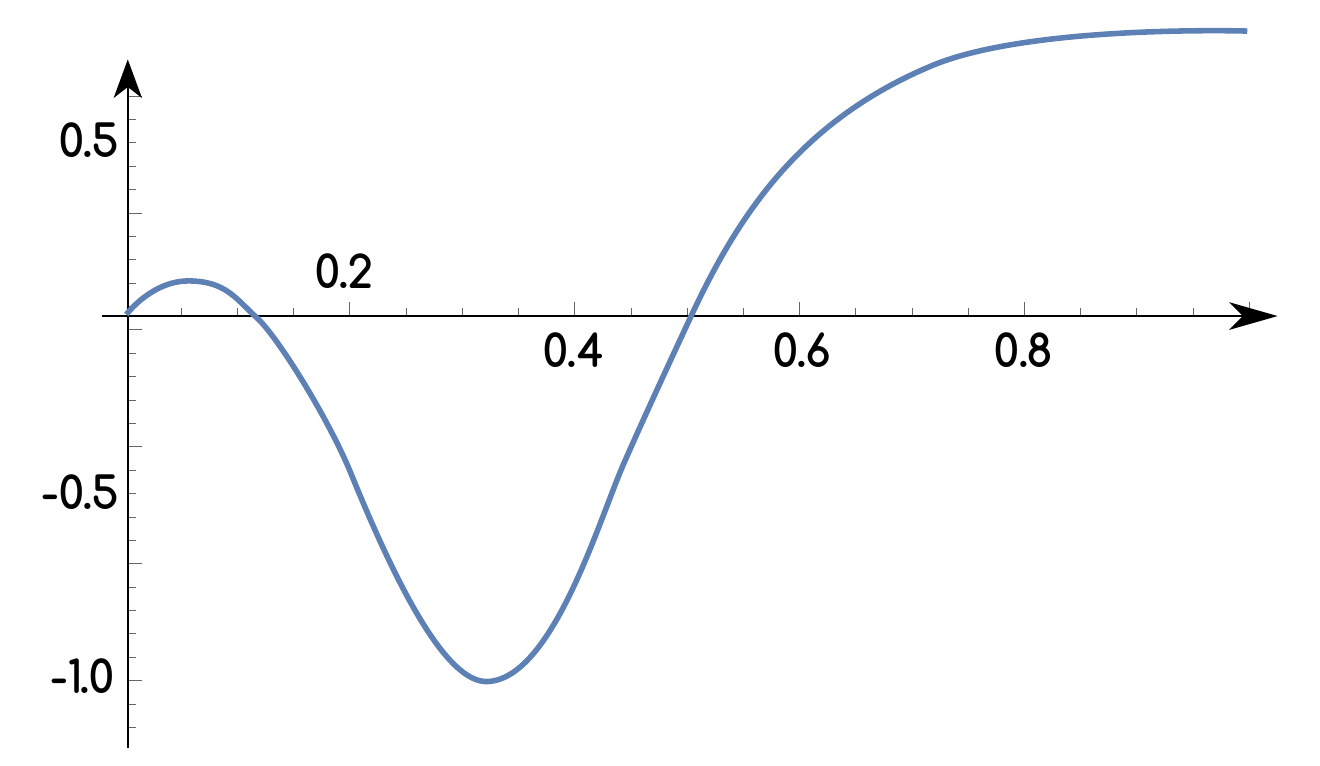}
\caption{Graph of $\sin\hat{\theta}$ (stationary point for $b=60$)}\label{ffffff}
\end{figure} 

\begin{figure}[H]
\centering
\includegraphics[scale=0.5]{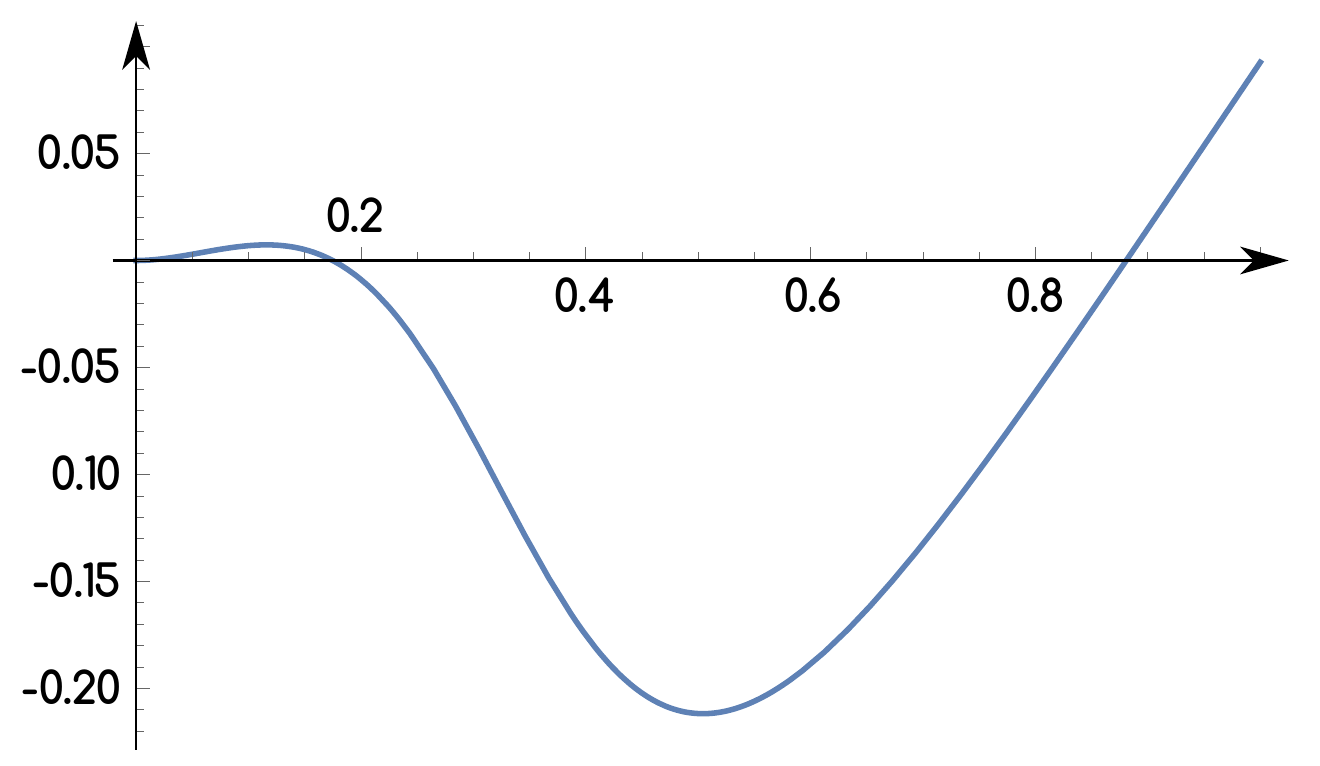} 
\caption{Graph of $\int_0^s\sin\hat{\theta}ds$}\label{gggggg}
\end{figure}

\noindent Summarizing, with $b=60$ (above the value $b\approx 41$ given in Fig.\ref{bifurcation}) we have numerical evidence of three solutions, two of which stable 
and one unstable.

\noindent The plot of the deformed shapes of the \textit{Elastica} relative to the three solutions $\bar\theta$, $\hat{\theta}$ and $\tilde\theta$ are shown in Figs. \ref{Stable}, \ref{deformatak1} and \ref{last}.

\begin{figure} [H]
\centering
\includegraphics[scale=0.3]{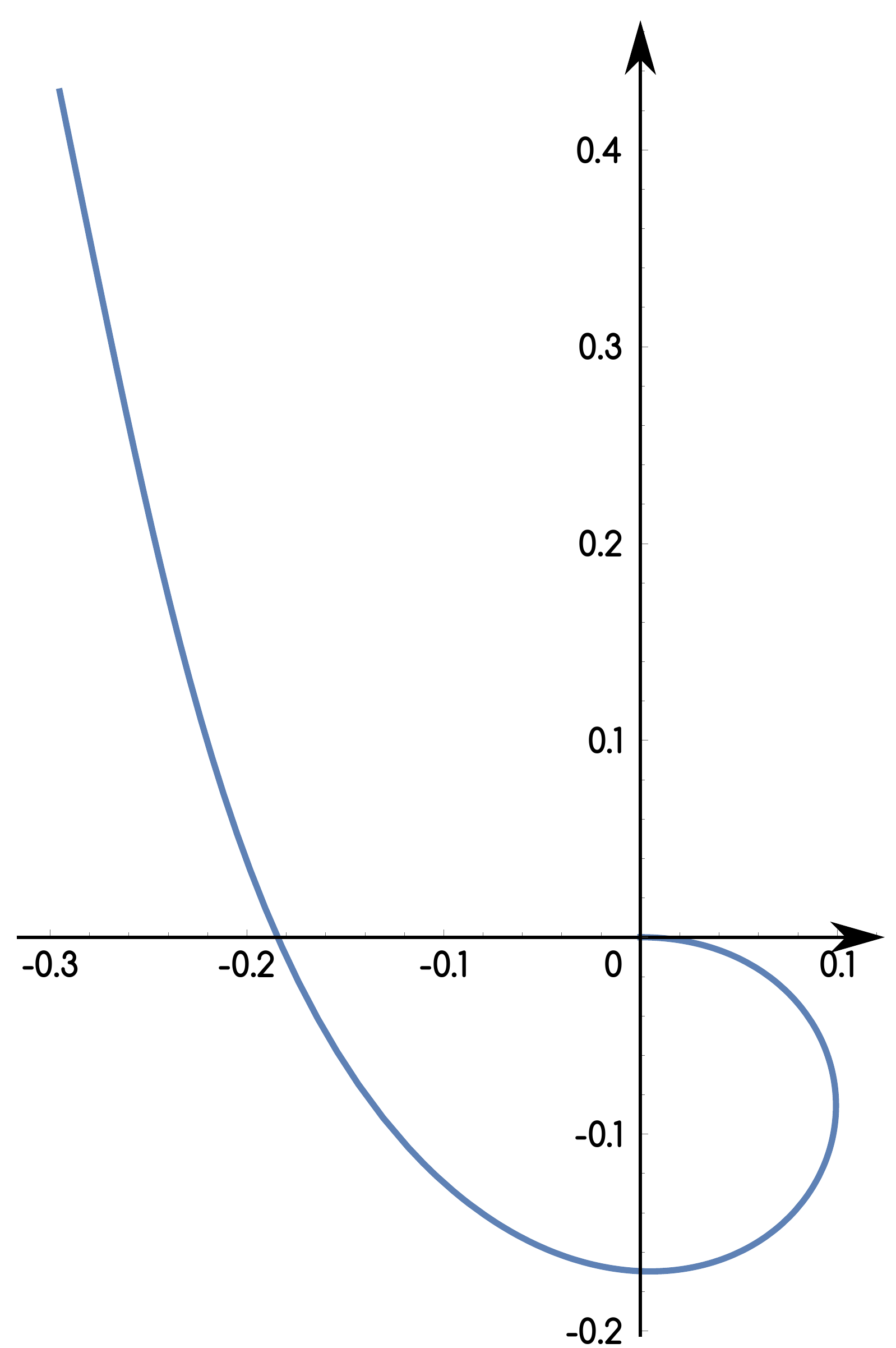}
\caption{Deformed shape of the \textit{Elastica} corresponding to $\bar\theta$}
\label{Stable}
\end{figure}

\begin{figure}[H]
\centering
\includegraphics[scale=0.3]{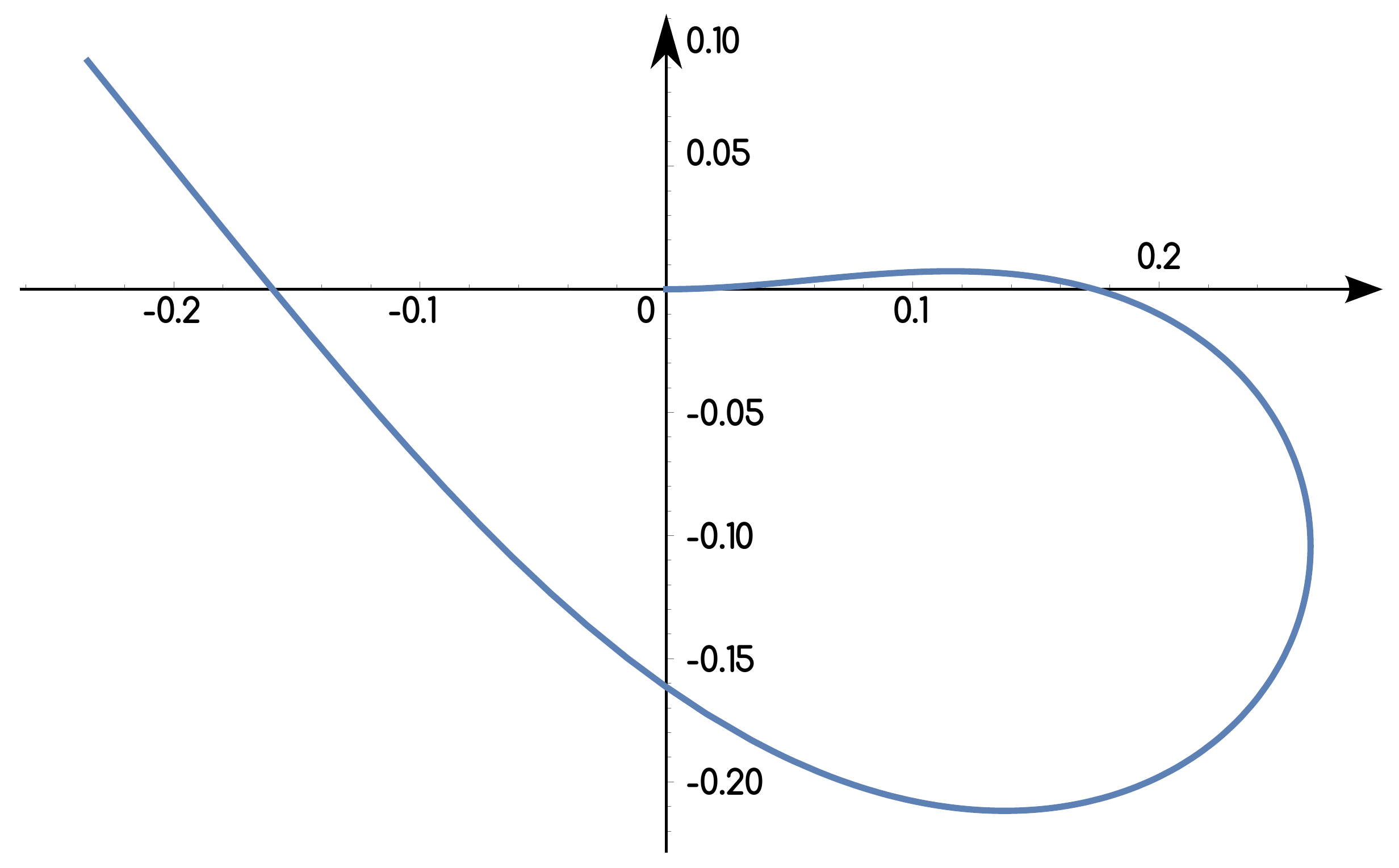} 
\caption{Deformed shape of the \textit{Elastica} corresponding to $\hat{\theta}$}\label{deformatak1}
\end{figure}

\begin{figure}[H]
\centering
\includegraphics[scale=0.4]{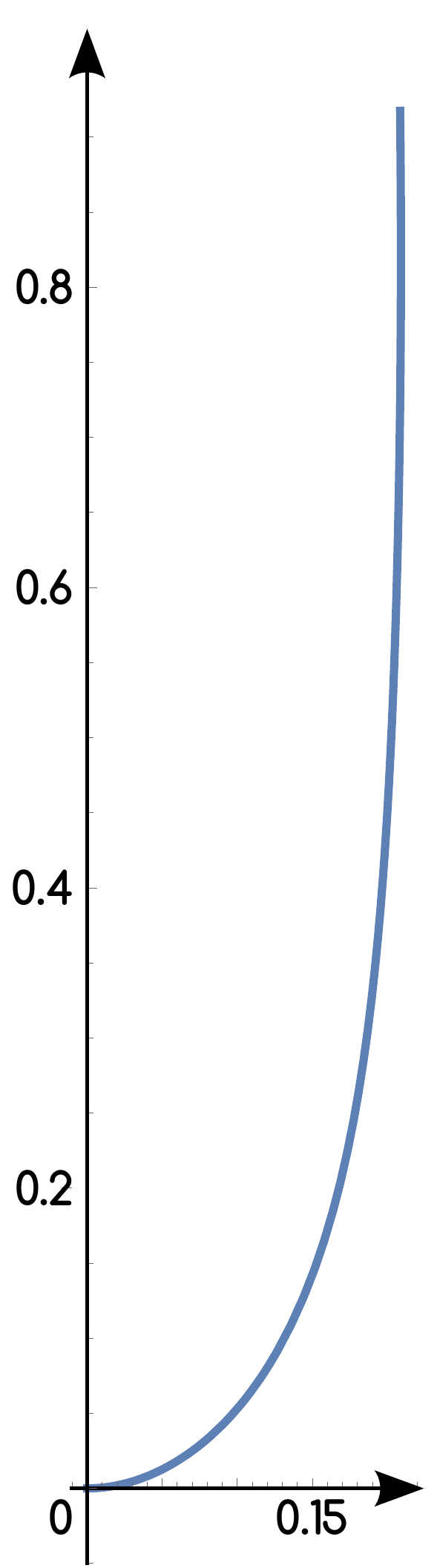} 
\caption{Deformed shape of the \textit{Elastica} corresponding to $\tilde\theta$}\label{last}
\end{figure}

\section*{7. Conclusions}

Direct method of calculus of variations provides a natural tool for the analysis of the classical problem of the \textit{Elastica} in the large deformation posed by Euler in 1744. In the present paper, we applied direct method  to establish the existence of global and local minimizers of a clamped \textit{Elastica} subjected to uniformly distributed load. We also gave sufficient conditions for stability or instability of particular classes of stationary configurations.  In particular, we proved that some equilibrium configurations found by means of numerical simulations are unstable if they exist. It seems to us that structural mechanics could significantly benefit from a wider application of the aforementioned theoretical tools. 

\noindent Several new questions have been opened within the present research, the final objective being the full characterization of the equilibria of Elastica in large deformation under a distributed load.

\vskip .5cm
\noindent\textbf{Acknowledgements.} We thank M. Ponsiglione for useful discussions.

\end{document}